\documentclass[12pt,journal,draftclsnofoot,onecolumn]{IEEEtran}
\usepackage{amsmath,amssymb,graphicx,epsfig,amsthm}

\newtheorem{definition}{Definition}
\newtheorem{proposition}{Proposition}
\usepackage{epstopdf}
\usepackage{balance}
\usepackage{multirow}
\usepackage{mathtools}
\usepackage{algorithm}
\usepackage{algpseudocode}
\usepackage{tabu}
\usepackage{subcaption}
\usepackage{cite}

\DeclareMathOperator*{\argmin}{arg\,min}
\DeclareMathOperator*{\argmax}{arg\,max}

\DeclarePairedDelimiter\floor{\lfloor}{\rfloor}

\newcommand{\pds}[1]{\ensuremath{\widetilde{#1}}}

\allowdisplaybreaks

\title{\LARGE Accelerated Structure-Aware Reinforcement Learning for \\ Delay-Sensitive Energy Harvesting Wireless Sensors}

\begin{document}

\author{Nikhilesh Sharma, Nicholas Mastronarde, and Jacob Chakareski}

\maketitle

%

\begin{abstract}
We investigate an energy-harvesting wireless sensor transmitting latency-sensitive  data over a fading channel.
The sensor injects captured data packets into its transmission queue and relies on ambient energy harvested from the environment to transmit them. 
We aim to find the optimal scheduling policy that decides whether or not to transmit the queue's head-of-line packet at each transmission opportunity such that the expected packet queuing delay is minimized given the available harvested energy.
No prior knowledge of the stochastic processes that govern the channel, captured data, or harvested energy dynamics are assumed,
thereby necessitating the use of online learning to optimize the scheduling policy.
We formulate this scheduling problem as a Markov decision process (MDP) and analyze the structural properties of its optimal value function. In particular, we show that it is non-decreasing and has increasing differences in the queue backlog and that it is non-increasing and has increasing differences in the battery state. 
We exploit this structure to formulate a novel accelerated reinforcement learning (RL) algorithm to solve the scheduling problem online at a much faster learning rate, while limiting the induced computational complexity. 
Our experiments demonstrate that the proposed algorithm closely approximates the performance of an optimal offline solution that requires a priori knowledge of the channel, captured data, and harvested energy dynamics. Simultaneously, by leveraging the value function's structure, our approach achieves competitive performance relative to a state-of-the-art RL algorithm, at potentially orders of magnitude lower complexity. Finally, considerable performance gains are demonstrated over the well-known and widely used Q-learning algorithm.
\end{abstract}



\section{Introduction}
\label{sec:intro}
Energy-constrained wireless sensors are increasingly used for latency-sensitive applications such as real-time remote visual sensing~\cite{Chakareski:15, Chakareski:11g}, the Internet of Things (IoT), body sensor networks~\cite{seyedi2010energy}, smart grid monitoring, and cyber-physical systems \cite{Chakareski:17, Chakareski:17a}. However, these sensors are subject to time-varying channel conditions and generate stochastic traffic loads
-- arising due to the compression algorithms that nodes apply to the sensed data before transmitting it~\cite{zordandesign} and due to the event-driven nature of many sensor network applications~\cite{seyedi2010energy, kansal2007power} --
which makes it very challenging for them 
to support latency-sensitive applications.
This is further complicated by the introduction of wireless sensors powered by energy harvested from the environment (e.g., ambient light, vibration/motion, or RF energy~\cite{vullers2010energy}). Although energy harvesting  sensors (EHSs) can operate autonomously in (remote) areas without access to power lines and without the need to change their batteries, the stochastic nature of harvested energy sources poses further challenges in sensor power management, transmission power allocation, and transmission scheduling due to the uncertainty in the amount of energy available for communication. Therein arises a need to study the behavior of scheduling policies employed by these sensors.

A lot of related work focuses on offline computation of optimal transmission policies for EHSs~\cite{gurakanenergy,lu2014dynamic, sharma2010optimal, gunduz2014designing}. 
For example, Gurakan and Ulukus~\cite{gurakanenergy} consider a multiaccess channel with two EHSs. Assuming that both energy and traffic arrive intermittently over time, and that their arrival processes are known a priori, they derive the optimal offline transmission power and rate allocations that maximize a sum rate objective function. Lu et al.~\cite{lu2014dynamic} formulate a throughput-optimal channel selection policy for EHSs operating as secondary users in a cognitive radio network. Gunduz et al.~\cite{gunduz2014designing} identify Markov decision processes (MDPs~\cite{puterman2014markov}) as a useful tool for optimizing EHSs in unpredictable environments with only causal information about the past and present, and statistical information about the future dynamics. Sharma et al.~\cite{sharma2010optimal} formulate both throughput-optimal and delay-optimal energy management policies as MDPs.
While these studies identify numerous techniques for calculating optimal transmission policies for EHSs offline, they do not provide analytical insights into the problems being studied and their structure.

Complementing the aforementioned research, another important body of work focuses on characterizing the structure of optimal transmission policies for EHSs~\cite{seyedi2010energy, ozel2011transmission, ho2012optimal, yang2012optimal1, yang2012optimal, michelusi2012optimal, aprem2013transmit, ho2010optimal}. For example, numerous studies have shown that optimal power allocation policies for EHSs have various water-filling structures~\cite{ozel2011transmission, ho2012optimal, yang2012optimal1}.
Ozel et al.~\cite{ozel2011transmission} consider two related problems: (i) maximizing the number of bits transmitted by a deadline and (ii) minimizing the time to transmit a certain number of bits. They identify that the transmission power over time that optimizes the first objective has a directional water-filling structure.
Ho and Zhang~\cite{ho2012optimal} consider the problem of throughput-optimal power allocation over a finite horizon. If unlimited energy can be stored in the battery and full state information is available about past, present, and future slots, they prove that the optimal energy allocation solution is based on water-filling, where the water levels follow a staircase function. Yang and Ulukus~\cite{yang2012optimal1} consider a two-user multiple access channel. Their goal is to minimize the required time by which all packets from both users are transmitted, by controlling the users' transmission powers and rates. Under the assumption that the energy harvesting times and amounts are known a priori, they prove that the optimal power allocation policy can be found by backward water-filling. 

Other types of structural results for EHSs are shown in~\cite{yang2012optimal, michelusi2012optimal, aprem2013transmit, zordandesign}.
For example, Yang and Ulukus~\cite{yang2012optimal} aim to adapt the transmission rate according to the traffic load and available energy, such that the time by which all packets are delivered is minimized. Assuming prior knowledge of the data and energy arrivals, they show that the optimal transmission rates increase in time.
Michelusi et al.~\cite{michelusi2012optimal} formulate the problem of maximizing the average importance of transmitted data as an MDP. They show that the EHS should only transmit data having an importance value above a certain threshold, which is a strictly decreasing function of the energy level.
Aprem et al.~\cite{aprem2013transmit} formulate outage optimal power control policies for EHSs. For the special case of binary power levels, they show that the optimal policy for the underlying MDP represents a threshold in the battery state. 
Zordan et al.~\cite{zordandesign} formulate optimal lossy compression policies for EHSs using constrained MDPs. They demonstrate that the optimal compression policy is non-decreasing in the battery, channel, and energy source states.

In practical scenarios, however, the stochastic processes governing the channel, captured data, and harvested energy dynamics are {\em unknown a priori}. This necessitates {\em online learning of transmission policies} to adapt on-the-fly to the experienced dynamics. In this context, reinforcement learning (RL),~\cite{sutton1998reinforcement,mastronarde2013joint}, has become an extremely useful tool. For instance, in \cite{blasco2013learning}, Blasco et al. propose the use of Q-learning~\cite{watkins1992q} (the most widely used RL technique) to maximize the throughput of an energy harvesting transmitter that cannot store the data in a buffer, i.e., the data is either transmitted in the time slot following its arrival or it is dropped. 
While Q-learning can solve problems with small action/state spaces, it exhibits very poor convergence rates. This makes it inappropriate for problems with large state spaces or tight timing constraints, such as the one we consider here. 

Other RL frameworks, e.g., SARSA, Bayesian RL, actor-critic learning~\cite{konda2000actor}, have also been very popular in the literature. Ortiz et al. \cite{ortiz2016reinforcement} use an approximate SARSA algorithm with linear function approximation in a point-to-point energy harvesting system with a finite battery to find a power allocation policy that aims at maximizing throughput. In \cite{xiao2015bayesian}, the authors propose a Bayesian RL approach in an energy harvesting system to decide the transmit power and the number of transmit data packets to maximize the long-term expected reward. In \cite{pandana2005near}, Pandana and Liu use an actor-critic algorithm with softmax action selection to compute an online policy that maximizes the average throughput subject to a total energy constraint, whereas, in \cite{liu2006rl}, Liu and Itamar propose an actor-critic based adaptive MAC protocol with $\epsilon$-action selection, where the nodes actively infer the state of other nodes using the RL based control mechanism. 

While the aforementioned work makes great progress towards demonstrating the utility of RL in communication systems, it solely considers data-driven RL algorithms that do not incorporate useful information from the underlying system model.
Exploiting such knowledge about the nature of the available actions (scheduling, routing, etc.), the system's dynamics (packet losses, queuing behavior, etc.), and the system's cost structure (energy, delay, etc.) can significantly increase the learning rate, decrease the complexity, and reduce the memory requirements of RL algorithms, thereby making them suitable for EHSs. We pursue this approach herein.

In particular, we exploit the structure of the problem at hand to investigate a novel accelerated RL framework based on value function approximation, which allows EHSs to learn near-optimal transmission policies online at a fast learning rate, while limiting the induced computational complexity.
Our specific contributions are as follows:
\begin{itemize}
\item We formulate the delay-sensitive energy harvesting scheduling (DSEHS) problem as an MDP that takes into account the stochastic captured data traffic loads, harvested energy, and channel dynamics. We propose an RL-based approach to solve it online without a priori knowledge of these dynamics.
\item We leverage so-called post-decision states (PDS) and virtual experience (VE) to accelerate the learning process. The former capture the system state once an action is taken, but before the unknown dynamics take place. The latter allows us to update the value function at multiple states in each time slot.
\item We show that the optimal value is non-decreasing and has increasing differences in the buffer state and that it is non-increasing and has increasing differences in the battery state. 
\item Based on these structural properties, we formulate a low-complexity structure-aware accelerated RL algorithm to solve the DSEHS problem. We demonstrate its ability to closely approximate the performance of an optimal offline policy calculated with a priori knowledge of the experienced dynamics. Simultaneously, we demonstrate that our approach achieves competitive performance to the state-of-the-art VE learning algorithm~\cite{mastronarde2013joint}, at potentially orders of magnitude lower computational complexity, and considerable performance gains over the well-known Q-learning algorithm.
\end{itemize}

The rest of the paper is organized as follows. We introduce our system model in Section~\ref{sec:sys}. 
We formulate the DSEHS problem in Section~\ref{sec:formulation}. We introduce our RL framework in Section~\ref{sec:learning}. We analyze the structural properties of the DSEHS problem in Section~\ref{sec:structural_properties} and formulate the proposed structure-aware accelerated RL algorithm in Section~\ref{sec:2d-grid-update}. We present our simulation results in Section~\ref{sec:sim} 
and conclude in Section~\ref{sec:conclusion}.

\section{Delay-Sensitive Energy-Harvesting Wireless Sensor Model}
\label{sec:sys}

We consider a time-slotted single-input single-output (SISO) point-to-point wireless communication system in which an energy harvesting sensor transmits latency-sensitive  data over a fading channel. The system model is depicted in Fig. 1. The system comprises two buffers: a packet buffer with size $N_b$ and an energy buffer (battery) with size $N_e$, where $N_b$ and $N_e$ are possibly infinite. We assume that time is divided into slots with length $\Delta T$ (seconds) and that the system's state in the $n$th time slot is denoted by $s^{n}\triangleq(b^{n},e^{n},h^{n})\in \mathcal{S}$, where $b^{n}\in \mathcal{S}_b=\left\{ 0,1,...,N_b\right\}$ is the packet buffer state (i.e., the number of backlogged data packets), $e^{n}\in \mathcal{S}_e=\left\{ 0,1,...,N_e\right\}$ is the battery state (i.e., the number of energy packets in the battery), and $h^{n} \in \mathcal{S}_h$ is the channel fading state. At the start of the $n$th time slot, the optimizer observes the state of the system and takes the binary scheduling action $a^{n} \in \mathcal{A} = \{0,1\}$, where $a^{n} = 1$ indicates that it transmits the head-of-line packet in the queue and $a^{n} = 0$ otherwise.

\textbf{Channel model:} We assume a block-fading channel that is constant during each time slot and may change from one slot to the next. Similar to earlier work \cite{salodkar2008line,mastronarde2013joint,zhang1999finite,ngo2010monotonicity,zordandesign}, we assume that the channel fading coefficient $h^{n}\in \mathcal{S}_h$ is known to the transmitter at the start of each time slot, that $\mathcal{S}_h$ denotes a finite set of $N_h$ channel states, and that the evolution of the channel state can be modeled as a finite state Markov chain with transition probability function $P^{h}(h^\prime|h)$.

\begin{figure}[!htb]
\centering
  \includegraphics[width=3.45in]{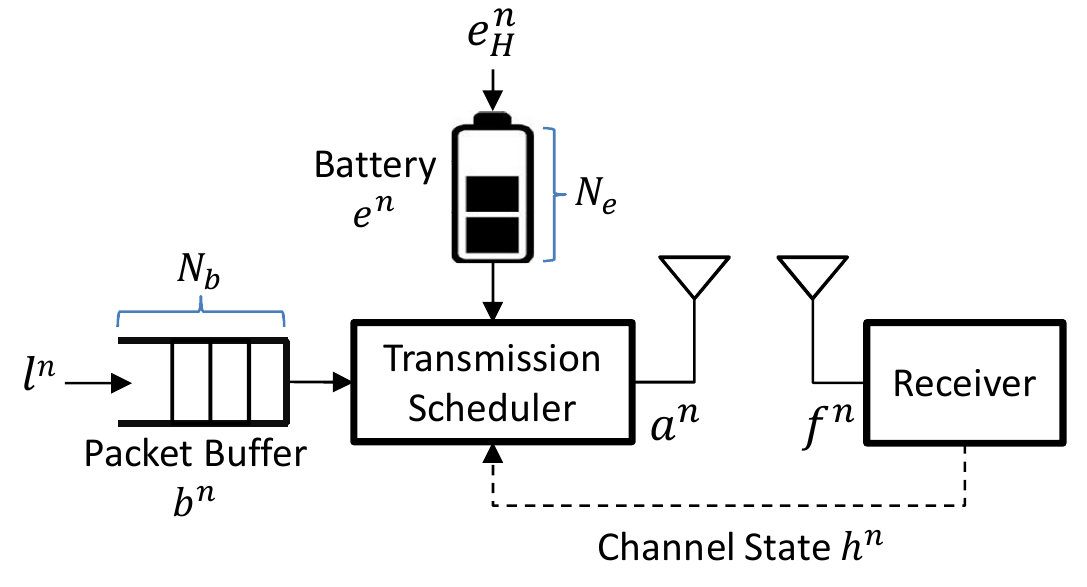}
  \caption{System block diagram.}
  \label{fig:1}
  \vspace{-0.5cm}
\end{figure}

\textbf{Energy harvesting model:} 
Similar to \cite{blasco2013learning}, we assume that battery energy is stored in the form of energy packets.
Let $e_{H}^{n} \in \mathcal{E} =\{0,1,\ldots,N_e\}$ denote the number of energy packets that are available for harvesting in the $n$th time slot and let $P^{e_H}(e_{H})$ denote the energy packet arrival distribution. Energy packets that arrive in time slot $n$ can be used in future time slots. Therefore, the battery state at the start of time slot $n+1$ can be found through the following recursion:
\begin{equation}\label{eq:e-recursion}
e^{n+1}=\min(e^{n}-e_{\text{TX}}(a^n)+e_{H}^{n},N_e),
\end{equation}
\noindent
where $e_{\text{TX}}(a^n)$ denotes the number of energy packets consumed in time slot $n$ given the scheduling action $a^n$. We assume that the wireless sensor uses a fixed transmission power $P_\text{TX}$ (energy packets per second); therefore,
\begin{equation}
	e_{\text{TX}}(a^n) = a^n P_\text{TX} \Delta T  = a^n e_{\text{TX}}\mbox{ (energy packets)}.
\end{equation}
For simplicity, we assume that the transmission energy $e_{\text{TX}}$ is an integer multiple of energy packets. Note that we only allow transmission actions $a^n$ such that $e_{\text{TX}}(a^{n}) \leq e^{n}$.

Given the energy packet arrival distribution $P^{e_H}(e_{H})$, the current state $s = (b,e,h)$, and the action $a$, the probability of observing battery state $e^\prime$ in the next time slot can be calculated as:
\begin{equation}\label{eq:battery-tpf}
	P^{e}(e^\prime|e,a) =
    	\sum_{e_H \in \mathcal{E}} \mathbb{I}_{\{e^\prime = \min(e-e_{\text{TX}}(a)+e_{H},N_e)\}}P^{e_H}(e_H),
\end{equation}
where $\mathbb{I}_{\{\cdot\}}$ is an indicator variable that is set to 1 when ${\{\cdot\}}$ is true and is set to 0 otherwise.



\textbf{Traffic model:}
Let $l^{n} \in \mathcal{L} = \{0,1,\ldots,M_l\}$ denote the number of data packets generated by the sensor in the $n$th time slot and let $P^l(l)$ denote the data packet arrival distribution. The buffer state at the start of time slot $n+1$ can be found through the following recursion:
\begin{equation}\label{eq:b-recursion}
b^{n+1}=\min(b^{n}-f^{n}(a^n,h^n)+l^{n},N_b),
\end{equation}
\noindent
where $f^{n}(a^n,h^n)$ is the number of packets transmitted successfully in time slot $n$ and $f^{n}(a^n,h^n) \leq a^{n} \leq b^n$. Note that new packet arrivals, and packets that are not successfully received, must be (re)transmitted in a future time slot. Assuming independent and identically distributed (i.i.d.) bit errors, we can characterize $f^n$ as a Bernoulli random variable with probability mass function
$P^{f}(f|a,h)$, where $P^{f}(f|0, h) = Bern(1)$ and $P^{f}(f|1, h) = Bern(q(h))$, 
where $q(h)$ is the packet loss rate (PLR) in channel state $h$. Since the transmission power is fixed, we have $q(h^+) < q(h^-)$ if $h^+ > h^-$.
We will refer to $P^f(f|a,h)$ as the goodput distribution. 

Given $P^{f}(f|a,h)$, the arrival distribution $P^l (l)$, the current state $s = (b,e,h)$, and the action $a$, the probability of observing buffer state $b^\prime$ in the next time slot can be calculated as follows:
\begin{equation}\label{eq:buffer-tpf}
	P^{b}(b^\prime|[b,h],a) =
    	\sum_{f \in \{0,1\}} \sum_{l \in \mathcal{L}}
        \mathbb{I}_{\{b^\prime = \min(b-f+l,N_b)\}}
        P^{f}(f|a,h)P^{l}(l),
\end{equation}
%
\section{The Delay-Sensitive Energy-Harvesting Scheduling (DSEHS) Problem}
\label{sec:formulation}

Let $\pi : \mathcal{S} \rightarrow \mathcal{A}$ denote a policy that maps states to actions. The objective of the delay-sensitive energy-harvesting scheduling (DSEHS) problem is to determine the optimal policy $\pi^*$ that minimizes the average packet queuing delay given the available energy. However, this does not mean that the policy should greedily transmit packets whenever there is enough energy to do so. On the contrary, it may be beneficial to abstain from transmitting packets in bad channel states and wait to transmit them in good channel states to reduce costly retransmissions and avoid wasting scarce harvested energy. At the same time, the policy should not be too conservative. For instance, if the battery is (nearly) full, transmitting a packet will make room for more harvested energy, which otherwise would have been lost due to the finite battery size.
To balance these considerations, we formulate the scheduling problem as an MDP~\cite{puterman2014markov}.

We define a \textit{buffer cost} to penalize large queue backlogs. Formally, we define the buffer cost as the sum of the \textit{holding cost} and the expected \textit{overflow cost} with respect to the arrival and goodput distributions, i.e., 
\begin{equation}\label{eq:cost}
  c([b,h],a)= \underbrace{b}_\text{holding cost} + \sum_{f \in \{0,1\}}\sum_{l \in \mathcal{L}}P^{l}(l)P^{f}(f|a,h) \biggl\{\underbrace{\eta\max(b-f+l-N_b,0)}_\text{overflow cost}\biggr\},
\end{equation}
The holding cost is nothing more than the buffer backlog, proportional to the queuing delay by Little's theorem~\cite{bertsekas1987data}. The overflow cost imposes a large penalty $\eta$ for each dropped packet. 

Stated formally, the DSEHS problem's objective is to determine the scheduling policy that solves the following optimization:
\begin{equation}
\begin{aligned}
& \underset{\pi \in \Pi}{\text{minimize}}
& &  \mathbb{E}\left[\sum\nolimits_{n=0}^{\infty}(\gamma)^{n}c(s^{n},\pi(s^{n}))\right],\\
\end{aligned}\label{eq:discounted_cost}
\end{equation}
where $\gamma \in [0, 1)$ is the discount factor, $\Pi$ is the set of all possible policies, and the expectation is taken over the sequence of states, which are governed by a controlled Markov chain with transition probabilities:
\begin{align}
P(s^\prime|s,a)=P^{b}(b^\prime|[b,h],a) \times P^{h}(h^\prime|h) \times P^{e}(e^\prime|e,a).\label{eq:tpf}
\end{align}
The optimal solution to~\eqref{eq:discounted_cost} satisfies the following Bellman equation, $\forall s \in \mathcal{S}$:
\begin{align}
V^{*}(s) &=& \min_{a \in \mathcal{A}(s)} & \biggl\{c(s, a)+\gamma\sum_{s^\prime \in \mathcal{S}}P(s^\prime | s, a)V^{*}(s^\prime)\biggr\},\label{eq:value} \\
&=& \min_{a\in \mathcal{A}(b, e)} & \biggl\{ c([b, h], a) + \gamma\sum_{l \in \mathcal{L}}\sum_{f \in \{0, 1\}} \sum_{e_H \in \mathcal{E}} \sum_{h^\prime \in \mathcal{S}_h} P^{l}(l) P^{f}(f | a, h) P^{e_H}(e_H)P^h(h^\prime | h)  \nonumber \\
&&& \quad  V^{*}([\min(b - f + l, N_b), \min(e - a \cdot e_\text{TX} + e_H, N_e), h^\prime]) \biggr\} \nonumber \\
&=& \min_{a \in \mathcal{A}(s)} & Q^{*}(s, a), \nonumber 
\end{align}

\noindent where,
\begin{equation}\label{eq:action-set}
  \mathcal{A}(b, e) = 
  \begin{cases}
    \{0, 1\}, & \mbox{if $b > 0$ and $e \geq e_{TX}$} \\
    \{0\}, & \mbox{otherwise,}
  \end{cases}
\end{equation}

\noindent is the set of feasible actions given the buffer and battery states, $V^{*}(s)$ is the optimal \textit{state-value function}, and $Q^{*}(s,a)$ is the optimal \textit{action-value function}. Then, the optimal policy $\pi^{*}(s)$ can be determined by taking the action in each state that minimizes the right-hand side of~\eqref{eq:value}.

Since the channel, energy arrival, and traffic arrival dynamics are unknown a priori, the optimal policy must be found using an online algorithm. Existing online approaches in energy harvesting systems typically rely on Q-learning~\cite{blasco2013learning}.
However, Q-learning exhibits extremely slow convergence rates for problems with many states and actions.
In our prior work \cite{mastronarde2013joint}, we proposed a fast RL algorithm that achieves three orders of magnitude faster convergence rates than Q-learning. However, its complexity is too high for EHSs. In Section~\ref{sec:approximate_RL}, we adapt the solution in~\cite{mastronarde2013joint} to create a fast and low-complexity RL algorithm based on value function approximation, which is better suited for EHSs. However, before we present the new algorithm, we must review the RL framework developed in~\cite{mastronarde2013joint}.

\section{Reinforcement Learning Framework}
\label{sec:learning}

In this section, we introduce fundamental RL concepts that we build on in Section~\ref{sec:2d-grid-update}. In Section~\ref{sec:pds-dp}, we review the concept of a post-decision state (PDS). In Section~\ref{sec:pds-learning}, we describe the PDS learning algorithm, which learns a value function defined over the PDSs. In Section~\ref{sec:convergence}, we prove that the PDS learning algorithm converges. In Section~\ref{sec:ve-learning}, we introduce the concept of virtual experience. 

\vspace{-6pt}
\subsection{Post-Decision State Based Dynamic Programming}\label{sec:pds-dp}

A PDS, denoted by $\pds{s}\triangleq(\pds{b},\pds{e},\pds{h})\in \mathcal{S}$, is a state of the system after all known dynamics have occurred, but before the unknown dynamics occur \cite{salodkar2008line,mastronarde2013joint,sutton1998reinforcement}. In the DSEHS problem, 
\begin{equation}\label{eq:pds}
\pds{s}^{n}=(\pds{b}^n,\pds{e}^n,\pds{h}^n)=([b^{n}-f^{n}],[e^{n}-a^{n} \cdot e_{\text{TX}}],h^{n})
\end{equation}
is the PDS in time slot $n$. 
The buffer's PDS $\pds{b}^n=b^{n}-f^{n}$ characterizes the buffer state after a packet is transmitted (if any), but
before any new packets arrive; the battery's PDS $\pds{e}^n=e^{n}-a^{n} \cdot e_{\text{TX}}$ characterizes the battery state after an energy packet is consumed (if any), but before any new energy packets arrive; and the channel's PDS $\pds{h}^{n} = h^{n}$ is the same as the channel state at time $n$. In other words, the PDS incorporates all of the known information about the transition from state $s^n$ to state $s^{n+1}$ after taking action $a^n$. Meanwhile, the unknown dynamics in the transition from state $s^{n}$ to $s^{n+1}$, i.e., the channel state transition from $h^n$ to $h^{n+1} \sim P^h(\cdot|h^{n})$, the data packet arrivals $l^{n} \sim P^l(\cdot)$, and the energy packet arrivals $e_{H}^{n} \sim P^{e_H}(\cdot)$ are not included in the PDS. The next state can be expressed in terms of the PDS as follows:
\begin{equation}
s^{n + 1} = (b^{n + 1}, e^{n + 1}, h^{n + 1}) = \left(\min(\pds{b}^n + l^n, N_b), \min(\pds{e}^n + e_H^n, N_e), h^{n + 1}\right).
\end{equation}

Just as we defined a value function over the conventional states, we can define a PDS value function over the PDSs. Let $\pds{V}^{*}$ denote the optimal PDS value function. $\pds{V}^{*}$ and the optimal value function $V^{*}$ are related by the following Bellman equations:
\begin{align}
\pds{V}^{*}(\pds{s}) =&~ \eta \sum\nolimits_{l \in \mathcal{L}}P^{l}(l)\max(\pds{b} + l - N_b, 0) + \nonumber \\
 &~ \gamma \sum_{l \in \mathcal{L}} \sum_{e_H \in \mathcal{E}} \sum_{h^\prime \in \mathcal{S}_h} P^{l}(l)P^{e_H}(e_H)P^h(h^\prime | h) V^{*}([\min(\pds{b} + l, N_b), \min(\pds{e} + e_H, N_e), h^\prime]) \label{eq:V_to_PDSV} \\
%
V^{*}(s) =&~ \min_{a \in \mathcal{A}(b, e)} \left\{b + \sum\nolimits_{f = 0}^{a} P^{f}(f | a, h)\pds{V}^{*}(b - f, e - a \cdot e_{TX}, h)\right\} \label{eq:PDSV_to_V}.
\end{align}
Knowing $\pds{V}^{*}(\pds{s})$, the optimal policy $\pi^{*}(s)$ can be found by taking the action in each state that minimizes the right-hand side of \eqref{eq:PDSV_to_V}.

\vspace{-6pt}
\subsection{Post-Decision State Learning}\label{sec:pds-learning}
PDS learning is a stochastic iterative algorithm for learning the PDS value function $\pds{V}^{*}(\pds{s})$ without prior knowledge of the data packet arrival distribution $P^l(l)$, energy packet arrival distribution $P^{e_H}(e_H)$, and channel transition probabilities $P^h(h^\prime | h)$. 

PDS learning is presented in Algorithm~\ref{alg:pds-learning}. At the start of time slot $n$, PDS learning takes the greedy action $a^{n}$ that minimizes the right-hand side of~\eqref{eq:pds-learning-greedy}. After observing the unknown dynamics (comprising the data packet arrivals $l^n \sim P^l(\cdot)$, energy packet arrivals $e^n_H \sim P^{e_H}(\cdot)$, and the next channel state $h^{n + 1} \sim P^h(\cdot | h^n)$), the algorithm evaluates the PDS $(\pds{b}^n, \pds{e}^n, \pds{h}^n)$ as defined in~\eqref{eq:pds}. 
The core of the PDS learning algorithm is the PDS value function update  defined in Algorithm~\ref{alg:update-PDSV} (\texttt{update\_PDSV}). When \texttt{update\_PDSV} is called in Algorithm~\ref{alg:pds-learning}, it takes as input the current PDS value function estimate $\pds{V}^{n}$, the current PDS $(\pds{b}^n, \pds{e}^n, \pds{h}^n)$, the current realization of the dynamics $(l^n, e_H^n, h^{n + 1})$, and the learning rate parameter $\beta^n \in [0, 1]$. It then uses~\eqref{eq:pds-update} to compute a new PDS value function estimate as a weighted average of (i) the current PDS value function estimate $\pds{V}^n(\pds{b}^n, \pds{e}^n, \pds{h}^n)$ and (ii) a new sample estimate of the PDS value function, i.e., $\eta \max(\pds{b}^n + l^n - N_b, 0) + \gamma V^{n}(b^{n + 1}, e^{n + 1}, h^{n + 1})$, derived based on the observed dynamics and the next state's estimated value $V^{n}(b^{n + 1}, e^{n + 1}, h^{n + 1})$ as computed in~\eqref{eq:pds-evaluate-V}. 

\begin{algorithm}
\caption{Post-Decision State Learning}
\label{alg:pds-learning}
\begin{algorithmic}[1]
\State \textbf{initialize} $\pds{V}^0(\pds{b}, \pds{e}, \pds{h}) = 0$ for all $(\pds{b}, \pds{e}, \pds{h}) \in \mathcal{S}$

\For {time slot $n = 0, 1, 2, \ldots$}

\State Take the greedy action:
\begin{equation}
\label{eq:pds-learning-greedy}
a^n = \argmin_{a \in \mathcal{A}(b^n, e^n)} \left\{b^n + \sum\nolimits_{f = 0}^{a} P^{f}(f | a, h^n)\pds{V}^n(b^n - f, e^n - a \cdot e_{TX}, h^n)\right\}
\end{equation}

\State Observe the data arrivals $l^n$, energy arrivals $e_H^n$, and next channel state $h^{n + 1}$

\State Evaluate the buffer's PDS $\pds{b}^n$, battery's PDS $\pds{e}^n$, and channel's PDS $\pds{h}^n$ using~\eqref{eq:pds}

\State $\pds{V}^{n + 1}(\pds{b}^n, \pds{e}^n, \pds{h}^n) \leftarrow $ \texttt{update\_PDSV}$\bigl(\pds{V}^{n}, [\pds{b}^n, \pds{e}^n, \pds{h}^n], [l^n, e_H^n, h^{n + 1}], \beta^n\bigr)$
\Comment{Algorithm~\ref{alg:update-PDSV}}

\EndFor
\end{algorithmic}
\end{algorithm}
\vspace{-0.5cm}
\begin{algorithm}
\caption{Post-Decision State Value Function Update (\texttt{update\_PDSV})}
\label{alg:update-PDSV}
\begin{algorithmic}[1]
\State \textbf{input} $\pds{V}$, $[\pds{b}, \pds{e}, \pds{h}]$, $[l, e_H, h^\prime]$, and $\beta$

\State Evaluate next buffer state $b^\prime = \min(\pds{b} + l, N_b)$ and next battery state $e^\prime = \min(\pds{e} + e_H, N_e)$

\State Evaluate the next state's value:
\begin{equation}\label{eq:pds-evaluate-V}
V(b^\prime, e^\prime, h^\prime) = \min_{a \in \mathcal{A}(b^\prime, e^\prime)} \biggl\{b^\prime + \sum\nolimits_{f = 0}^{a} P^{f}(f | a, h^\prime) \pds{V}(b^\prime - f, e^\prime - a \cdot e_{TX}, h^\prime) \biggr\}
\end{equation}

\State Update the PDS value function using the information in steps 1 -- 3:
\begin{equation}\label{eq:pds-update}
\pds{V}(\pds{b}, \pds{e}, \pds{h})\leftarrow(1 - \beta)\pds{V}(\pds{b}, \pds{e}, \pds{h})+\beta[\eta \max(\pds{b} + l - N_b, 0) + \gamma V(b^\prime, e^\prime, h^\prime)]
\end{equation}

\State \textbf{return} $\pds{V}(\pds{b}, \pds{e}, \pds{h})$
\end{algorithmic}
\end{algorithm}
\vspace{-0.7cm}

\subsection{The Convergence of Post-Decision State Learning}
\label{sec:convergence}
In this section, we prove that the sequence of PDS value functions $\widetilde{V}^n$ generated by the PDS learning algorithm converges to $\widetilde{V}^*$ with probability 1 as $n \rightarrow \infty$. We begin by introducing the concept of a ``well-behaved'' stochastic iterative algorithm, which is known to converge under mild conditions~\cite{bertsekas1995neuro}. In the remainder of this section, we let $\|X\|$ denote the $L_\infty$ norm of the vector $X$, i.e., $\|(X(1), X(2), \ldots, X(k) )\| = \max_i X(i)$.

Consider a stochastic iterative algorithm with the following form:
\begin{equation}\label{eq:stoch-iter}
X^{n + 1}(i) = (1 - \beta^n) X^{n}(i) + \beta^n [(H^n X^{n})(i) + w^n(i)],
\end{equation}
where $w^n$ is a bounded random variable with zero expectation and $H^n$
belongs to a family of contraction mappings. The iteration in~\eqref{eq:stoch-iter} constitutes a well-behaved stochastic algorithm if it satisfies the following conditions:

\begin{definition}
\label{def:well-behaved}
(Well-behaved stochastic iterative algorithm~\cite{bertsekas1995neuro}): A stochastic iterative algorithm is well-behaved if:
\begin{enumerate}
\item Stochastic approximation conditions: The non-negative step sizes $\beta^n$ satisfy $\sum_{n = 0}^\infty \beta^n = \infty$ and $\sum_{n = 0}^\infty (\beta^n)^2 \leq \infty$.

\item Bounded noise: There exists a constant $G$ that bounds $w^n(i)$ for any history $F^n$, i.e., $|w^n(i)| \leq G, \forall n, i$.

\item Contraction mapping: There exists a $\gamma \in [0, 1)$ and a vector $X^*$ such that for any $X$ we have $||H^n X - X^*|| \leq \gamma||X - X^*||$.
\end{enumerate}
\end{definition}

\begin{proposition}\label{prop:well-behaved}
The PDS learning algorithm defined in Algorithm~\ref{alg:pds-learning} is a well-behaved stochastic iterative algorithm.
\end{proposition}
\begin{proof}
The proof is given in the appendix.
\end{proof}

Note that, although PDS learning converges, it does so relatively slowly because it only updates the value of one PDS in each time slot. In the next subsection, we introduce the concept of \textit{virtual experience}, which allows us to update multiple PDSs in each time slot thereby dramatically improving the convergence rate.

\vspace{-6pt}
\subsection{Virtual Experience Learning}
\label{sec:ve-learning}
Virtual experience learning is a state-of-the-art reinforcement learning algorithm that we proposed in our prior work~\cite{mastronarde2013joint}. The key idea behind virtual experience learning is that it is possible to update the value of multiple PDSs in each time slot. In the DSEHS problem, virtual experience learning is enabled by the fact that the unknown data arrival, energy packet arrival, and channel transition dynamics (i.e., $l^n \sim P^l(l)$, $e_H^n \sim P^{e_H}(e_H)$, and $h^{n+1} \sim P^{h}(h'|h)$, respectively) are independent of the post-decision buffer and battery states (i.e., $\pds{b}^n$ and $\pds{e}^n$, respectively). This enables us to update all PDSs with the same $\pds{h}^n$, but with different $\pds{b}$ and $\pds{e}$ given the observations of $l^n$, $e^{n}_H$, and $h^{n + 1}$. Updating $|\mathcal{S}_b \times \mathcal{S}_e|$ PDSs in every time slot significantly improves the convergence rate at the cost of increased computational complexity. Specifically, if the update is applied every $T$ time slots, then the average number of PDSs updated in each time slot is $|\mathcal{S}_b \times \mathcal{S}_e|/T$. Algorithm~\ref{alg:ve-learning} provides pseudo-code for the virtual experience learning algorithm with an update period $T = 1$.


\begin{algorithm}[!htb]
\caption{Virtual Experience Learning (update period $T=1$)}
\label{alg:ve-learning}
\begin{algorithmic}[1]
\State \textbf{initialize} $\pds{V}^0(\pds{b}, \pds{e}, \pds{h}) = 0$ for all $(\pds{b}, \pds{e}, \pds{h}) \in \mathcal{S}$

\For {time slot $n = 0, 1, 2, \ldots$}

\State Take the greedy action:
\begin{equation}
\label{eq:pds-greedy}
a^n = \argmin_{a \in \mathcal{A}(b^n, e^n)} \left\{b^n + \sum\nolimits_{f = 0}^{a} P^{f}(f | a, h^n)\pds{V}^n(b^n - f, e^n - a \cdot e_{TX}, h^n)\right\}
\end{equation}

\State Observe data arrivals $l^n$, energy arrivals $e_H^n$, and next channel state $h^{n + 1}$

\For {all $(\pds{b}, \pds{e}) \in \mathcal{S}_b \times \mathcal{S}_e$}

\State $\pds{V}^{n + 1}(\pds{b}, \pds{e}, h^n) \leftarrow $ \texttt{update\_PDSV}$\bigl(\pds{V}^{n}, [\pds{b}, \pds{e}, h^n], [l^n, e_H^n, h^{n + 1}], \beta^n\bigr)$
\Comment{Algorithm~\ref{alg:update-PDSV}}

\EndFor

\EndFor
\end{algorithmic}
\end{algorithm}
\section{Value Function Approximation-Based Reinforcement Learning}\label{sec:approximate_RL}

The virtual experience learning algorithm is too complex to implement on EHSs because it requires updating $|\mathcal{S}_b \times \mathcal{S}_e|$ PDSs every update period of $T$ time slots. Although $T$ can be increased to further reduce the average learning complexity per time slot, this comes at the expense of a significant decrease in the convergence rate~\cite{mastronarde2013joint}. 

In this section, we pursue a more effective approach to reduce the complexity of virtual experience learning, while still reaping its benefits. Specifically, we propose to learn an \textit{approximate} value function instead of the true value function. To this end, we first present several structural properties of the optimal PDS value function $\pds{V}^{*}(s)$ (Section~\ref{sec:structural_properties}). Then, motivated by these properties, we propose a novel RL algorithm that learns a near-optimal piece-wise planar approximation of the PDS value function (Section~\ref{sec:2d-grid-update}).


\vspace{-6pt}
\subsection{Structural Properties of the Optimal Value Function}
\label{sec:structural_properties}

Integer convexity is key to understanding the structure of the optimal PDS value function.

\begin{definition}
(Integer Convex): An integer convex function $f(n): \mathcal{N} \rightarrow \mathbb{R}$ on a set of integers $\mathcal{N} \in \{0, 1, \ldots ,N\}$ is a function that has increasing differences in $n$, i.e., 
\begin{equation}
	f(n_1 + m) - f(n_1) \leq f(n_2 + m) - f(n_2)
\end{equation}
for $n_1 < n_2$ and $n_1, n_2, n_1 + m, n_2 + m \in \mathcal{N}$. 
\end{definition} 

The following propositions establish the key structural properties of the PDS value function with respect to the post-decision buffer state $\pds{b}$ and the post-decision battery state $\pds{e}$, respectively. The proofs are omitted due to space limitations, but can be found in~\cite{sharma2018structural}.

\begin{proposition}\label{prop:structure-PDSV-b}
The optimal PDS value function $\pds{V}^*(\pds{b}, \pds{e}, \pds{h})$ has the following structural properties with respect to the post-decision buffer state $\pds{b}$:
\begin{enumerate}
\item $\pds{V}^*(\pds{b}, \pds{e}, \pds{h})$ is non-decreasing in the post-decision buffer state \pds{b}, i.e.,
\begin{equation} \label{eq:V-non-decr-b}
\pds{V}^*(\pds{b}, \pds{e}, \pds{h}) \leq \pds{V}^*(\pds{b}+1, \pds{e}, \pds{h}).
\end{equation}
\item If the packet buffer has infinite size ($N_b = \infty$), then $\pds{V}^*(\pds{b}, \pds{e}, \pds{h})$ has increasing differences in the post-decision buffer state $\pds{b}$, i.e.,
\begin{equation}\label{eq:incr-diff-PDSV-b}
\pds{V}^*(\pds{b}, \pds{e}, \pds{h}) - \pds{V}^*(\pds{b} - 1, \pds{e}, \pds{h}) \leq \pds{V}^*(\pds{b} + 1, \pds{e}, \pds{h}) - \pds{V}^*(\pds{b}, \pds{e}, \pds{h}).
\end{equation}
\end{enumerate}
\end{proposition}

\begin{proposition}\label{prop:structure-PDSV-e}
The optimal PDS value function $\pds{V}^*(\pds{b},\pds{e},\pds{h})$ has the following structural properties with respect to the post-decision battery state \pds{e}:
\begin{enumerate} 
\item $\pds{V}^*(\pds{b}, \pds{e}, \pds{h})$ is non-increasing in the post-decision battery state $\pds{e}$, i.e.,
\begin{equation} \label{eq:V-non-incr-e}
\pds{V}^*(\pds{b}, \pds{e}, \pds{h}) \geq \pds{V}^*(\pds{b}, \pds{e}+1, \pds{h}).
\end{equation}
\item $\pds{V}^*(\pds{b},\pds{e},\pds{h})$ has increasing differences in the post-decision battery state $\pds{e}$, i.e.,
\begin{equation}\label{eq:incr-diff-PDSV-e}
\pds{V}^*(\pds{b}, \pds{e}, \pds{h}) - \pds{V}^*(\pds{b}, \pds{e} - 1, \pds{h}) \leq \pds{V}^*(\pds{b}, \pds{e} + 1, \pds{h}) - \pds{V}^*(\pds{b}, \pds{e}, \pds{h}).
\end{equation}
\end{enumerate}
\end{proposition}
 
Proposition~\ref{prop:structure-PDSV-b} implies that the cost to serve an additional data packet increases with the queue backlog. In~\cite{sharma2018structural}, we were only able to prove that $\pds{V}^*(\pds{b}, \pds{e}, \pds{h})$  has increasing differences in the buffer state for an infinite size buffer; however, we have not observed any cases in practice where this property does not hold for finite buffers. 
Proposition~\ref{prop:structure-PDSV-e} implies that the benefit of an additional energy packet decreases with the available battery energy. 


\vspace{-6pt}
\subsection{Grid Learning}
\label{sec:2d-grid-update}
Since the optimal PDS value function has increasing differences in the post-decision buffer and battery states (see Propositions~\ref{prop:structure-PDSV-b} and~\ref{prop:structure-PDSV-e}), we propose to approximate it as a piece-wise planar function. Using this approximation, we develop an adaptive low-complexity reinforcement learning algorithm that can quickly learn an approximation of the optimal PDS value function with bounded and controllable error. We refer to this structure-aware algorithm as \textit{grid learning}.

For each post-decision channel state $\pds{h} \in \mathcal{S}_h$, the grid learning algorithm constructs a two-dimensional grid of post-decision buffer and battery states on which to learn the PDS value function. Rather than using a uniform grid, however, we propose to use a \textit{quadtree} data structure so that our value function approximation can be adaptively refined in space and time (i.e., on the buffer-battery plane and from slot-to-slot) to meet a predetermined approximation error tolerance, $\delta$. Each leaf of the quadtree is then divided into two triangles, which lie on two intersecting planes. Together, the planes of all leaf nodes compose the proposed piece-wise planar approximation.

The remainder of this subsection is organized as follows. In Section~\ref{sec:quadtree}, we formalize the quadtree data structure and present relevant quadtree operations. In Section~\ref{sec:grid-alg}, we present pseudocode for the grid learning algorithm. In Section~\ref{sec:appoximation}, we discuss how the value function can be approximated from the quadtree. Finally, we describe how we adaptively refine the quadtree to meet the target error tolerance in Section~\ref{sec:grid-update}.

\subsubsection{Quadtree definition}\label{sec:quadtree}
Let $\mathcal{T}$ denote a quadtree defined on the set of buffer-battery state pairs $\mathcal{S}_b \times \mathcal{S}_e$ within a bounding box (\texttt{BB}) defined as follows (cf. Fig.~\ref{fig:quadtree-bb}):
\begin{equation} \label{eq:bounding-box}
\texttt{BB}(\mathcal{T}) = \{(b_{-},e_{-}), (b_{+},e_{-}), (b_{-},e_{+}), (b_{+}, e_{+})\}, 
\end{equation}
where $0 \leq b_{-} < b_{+} \leq N_b$ and $0 \leq e_{-} < e_{+} \leq N_e$. In words, $\texttt{BB}(\mathcal{T})$ comprises the extreme vertices of the quadtree. 
We say that $(b,e)$ lies inside $\mathcal{T}$'s bounding box if $b_{-} \leq b \leq b_{+}$ and $e_{-} \leq e \leq e_{+}$; otherwise, $(b,e)$ lies outside of $\mathcal{T}$'s bounding box. 

If $\mathcal{T}$ is a leaf node, then it can be subdivided into four sub-quadtrees (children) spanning its northwest (NW), northeast (NE), southwest (SW), and southeast (SE) quadrants, i.e., $\texttt{subdivide}(\mathcal{T}) = \{\mathcal{T}_{NW},\mathcal{T}_{NE},\mathcal{T}_{SW},\mathcal{T}_{SE}\}$, with bounding boxes defined as follows (cf. Fig.~\ref{fig:quadtree-subdivide}):
\begin{align*}
\texttt{BB}(\mathcal{T}_{NW}) =& \{(b_{-},\bar{e}), (\bar{b},\bar{e}), (b_{-},e_{+}), (\bar{b}, e_{+})\}, \quad
\texttt{BB}(\mathcal{T}_{NE}) =& \{(\bar{b},\bar{e}), (b_{+},\bar{e}), (\bar{b},e_{+}), (b_{+}, e_{+})\}, \\
\texttt{BB}(\mathcal{T}_{SW}) =& \{(b_{-},e_{-}), (\bar{b},e_{-}), (b_{-},\bar{e}), (\bar{b},\bar{e})\}, \quad
\texttt{BB}(\mathcal{T}_{SE}) =& \{(\bar{b},e_{-}), (b_{+},e_{-}), (\bar{b},\bar{e}), (b_{+}, \bar{e})\},
\end{align*}
where $\bar{b} = \floor{\frac{b_{+}+b_{-}}{2}} \in \mathcal{S}_b$, $\bar{e} = \floor{\frac{e_{+}+e_{-}}{2}} \in \mathcal{S}_e$, and $\floor{x}$ is the floor operator, which denotes the largest integer that is smaller than $x$. 
With a slight abuse of notation, we write $(b,e) \in \mathcal{T}$ if $(b,e)$ is an element of $\mathcal{T}$'s bounding box or one of its children's bounding boxes, recursively down to all of its leaf nodes.

\begin{figure}[!htb]
	\centering
    \begin{subfigure}{.45\textwidth}
    	\centering
    	\includegraphics[width=1\textwidth,trim={1cm 5cm 9cm 1cm},clip]{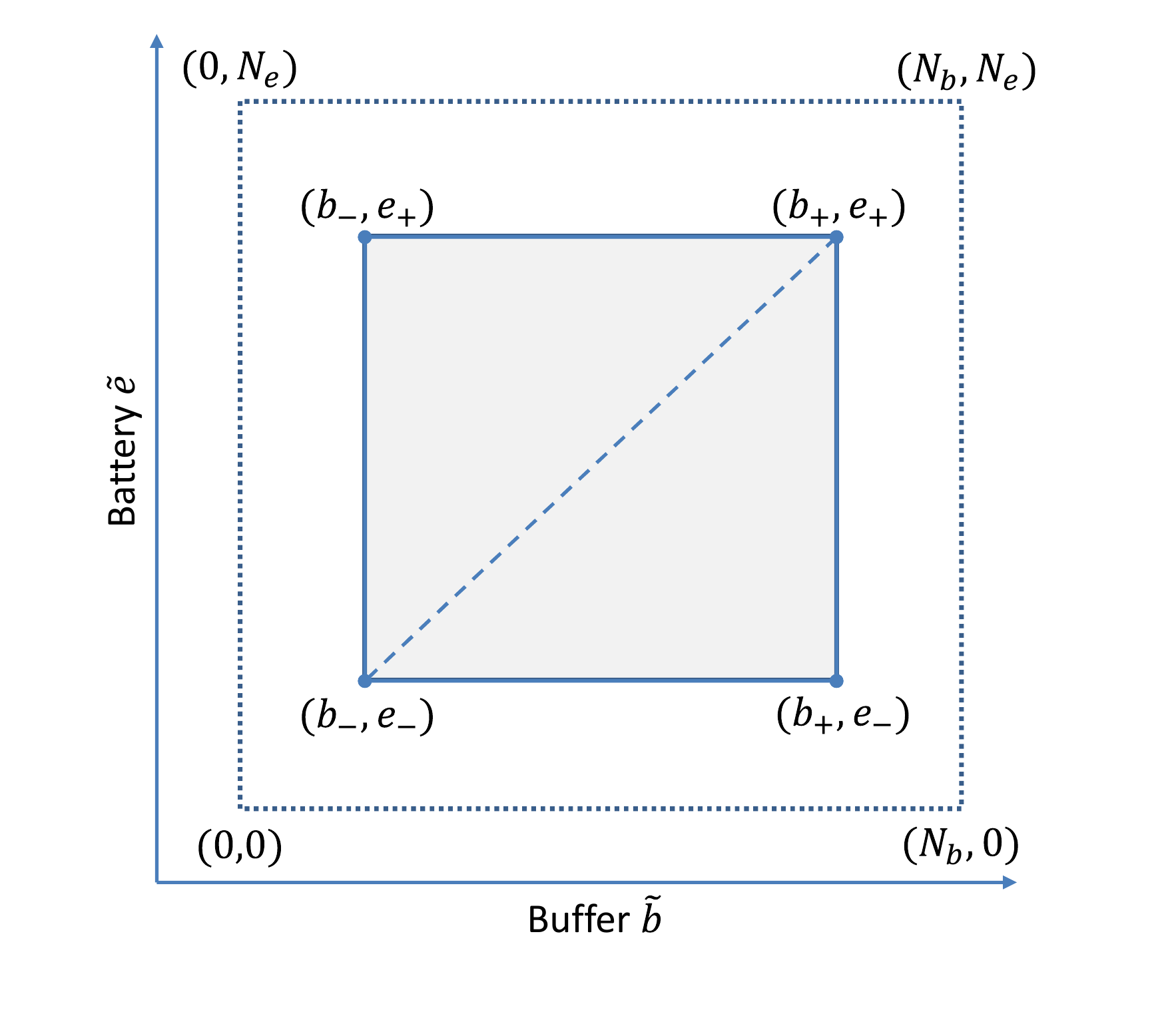}
    	\subcaption{Quadtree bounding box.}
        \label{fig:quadtree-bb}
  	\end{subfigure}
	\begin{subfigure}{.45\textwidth}
    	\centering
    	\includegraphics[width=1\textwidth,trim={1cm 5cm 9cm 1cm},clip]{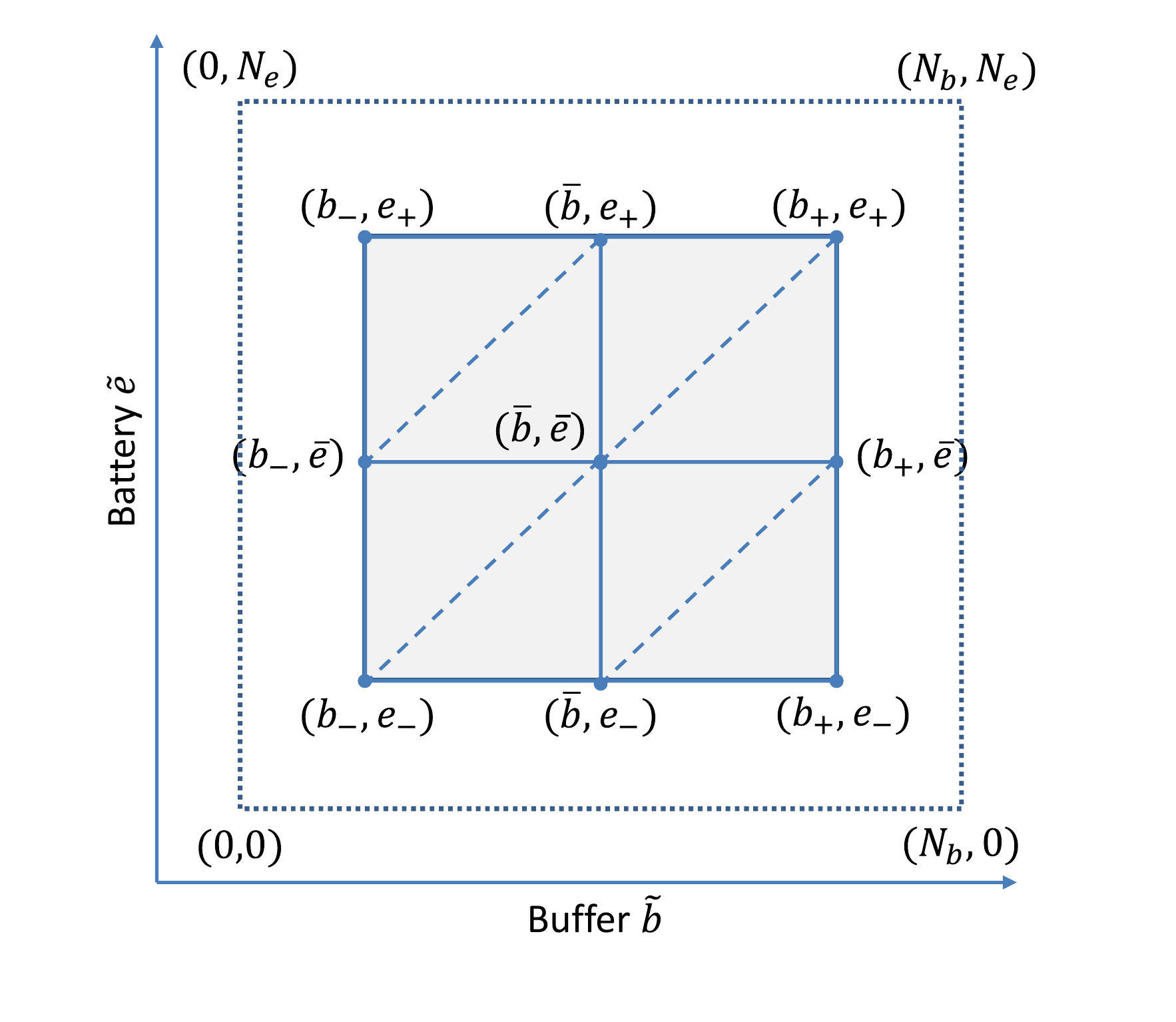}
    	\subcaption{Quadtree subdivide operation.}
        \label{fig:quadtree-subdivide}
  	\end{subfigure}
\caption{Quadtree construction. Each leaf node is divided into a northwest and southeast triangle as part of the piece-wise planar approximation.}
\label{fig:quadtree-def}
\end{figure}

\subsubsection{Grid learning algorithm} \label{sec:grid-alg}
Let $\mathcal{T}^{n}(\pds{h})$ denote the quadtree used to approximate the PDS value function in channel state $\pds{h}$ in time slot $n$. We assume that $\texttt{BB}(\mathcal{T}^{n}(\pds{h}))$ is defined as in \eqref{eq:bounding-box} for all $n$. 
Note that we do not require $\mathcal{T}^n(\pds{h})$ to span the entire buffer-battery plane (i.e., for $b_{-}=0$, $b_{+} = N_b$, $e_{-}=0$, and $e_{+} = N_e$) because $N_b$ and $N_e$ may be very large (or infinite) and it is often unnecessary to accurately approximate the value at the extremes of the state space (e.g., if there is an abundant supply of energy or very little data to serve). 

The grid learning algorithm approximates the value of any PDS pair $(\pds{b}, \pds{e}) \notin \mathcal{T}^{n}(\pds{h})$ using the values of PDS pairs $(\pds{b}, \pds{e}) \in \mathcal{T}^{n}(\pds{h})$. That is, instead of operating directly on the PDS value function $\pds{V}$, it operates on an approximate PDS value function $\hat{V}$ such that
\begin{equation} \label{eq:approx-PDSV}
\hat{V}^n(\pds{b}, \pds{e}, \pds{h}) =
\begin{cases}
\pds{V}^n(\pds{b}, \pds{e}, \pds{h}), & \mbox{if } (\pds{b}, \pds{e}) \in \mathcal{T}^{n}(\pds{h}) \\
\text{\texttt{approximate\_PDSV}}\bigl(\pds{V}^n, \mathcal{T}^{n}(\pds{h}), [\pds{b}, \pds{e}, \pds{h}]\bigr), & \mbox {otherwise.}
\end{cases}
\end{equation}
In Section~\ref{sec:appoximation}, we describe how the function \texttt{approximate\_PDSV} calculates the approximate value of buffer-battery state pairs that lie inside or outside of $\mathcal{T}^n(\pds{h})$'s bounding box.

Pseudocode for the grid learning algorithm with update period $T = 1$ is provided in Algorithm~\ref{alg:grid-learning}. At the start of the algorithm ($n = 0$), we initialize $\mathcal{T}^{0}(\pds{h})$ with $\texttt{BB}(\mathcal{T}^{0}(\pds{h}))$ defined as in \eqref{eq:bounding-box} and initialize its child nodes to empty. In other words, $\mathcal{T}^{0}(\pds{h})$ serves as the root of the quadtree and provides the minimum set of grid points from which we can estimate the values of all $(\pds{b}, \pds{e}) \in \mathcal{S}_b \times \mathcal{S}_e$ using the proposed piece-wise planar approximation. 
After initialization, the algorithm proceeds similarly to virtual experience learning (Algorithm~\ref{alg:ve-learning}) with three key differences. 
First, as noted above, the algorithm operates on an approximate PDS value function $\hat{V}$ instead of the actual PDS value function $\pds{V}$.\footnote{In Algorithm~\ref{alg:grid-learning}, we slightly abuse the notation when we use $\hat{V}^{n}$ on the right-hand side of~\eqref{eq:grid-greedy} and as an argument to the \texttt{update\_PDSV} function. In practice, we have chosen to calculate values of $\hat{V}^{n}$ on-demand using the \texttt{approximate\_PDSV} function. In this way, we do not need to maintain a full tabular representation of the (approximate) value function.} 
Second, the function \texttt{update\_PDSV} is only called for PDS pairs $(\pds{b}, \pds{e}) \in \mathcal{T}^{n}(\pds{h})$, rather than all PDS pairs $(\pds{b}, \pds{e}) \in \mathcal{S}_b \times \mathcal{S}_e$. Since $\mathcal{T}^{n}(\pds{h})$ is only a small subset of $\mathcal{S}_b \times \mathcal{S}_e$ and $\pds{V}$ is only defined on $\mathcal{T}^{n}(\pds{h}), \forall \pds{h} \in \mathcal{S}_h$, the grid learning algorithm requires significantly less computation and memory than exhaustive virtual experience learning operating on the full PDS value function (i.e., Algorithm~\ref{alg:ve-learning}). 
Third, since the approximate value function $\hat{V}$ may not approximate all PDS pairs $(\pds{b}, \pds{e}) \in \mathcal{S}_b \times \mathcal{S}_e$ within the target error tolerance $\delta$, we use the \texttt{update\_grid} function (Algorithm~\ref{alg:DynamicGrid}) to adaptively refine the approximation over time. 
We now describe the \texttt{approximate\_PDSV} and \texttt{update\_grid} functions in detail.

\subsubsection{PDS value function approximation} \label{sec:appoximation}
Suppose that $\mathcal{T}(h)$ is the root of the quadtree and that we want to find the  approximate value $\hat{V}(b,e,h)$ of the buffer-battery state pair $(b,e)$, which may or may not lie inside of $\mathcal{T}(h)$'s bounding box as defined in \eqref{eq:bounding-box}.
The function \texttt{approximate\_PDSV} achieves this in roughly four steps: 1) associate $(b,e)$ with one of the quadtree's leaf nodes; 2) further associate $(b,e)$ with the leaf node's NW or SE triangle; 3) find the equation of the  plane defined by the selected triangle's vertices (hereafter, we will refer to this as the \textit{approximating plane}); and 4) calculate the approximate value of $\hat{V}(b,e,h)$ from the approximating plane.

To be precise, we first associate $(b,e)$ with the quadtree's nearest leaf node using a recursive search from the root. Subsequently, we associate $(b,e)$ with the leaf node's nearest triangle as illustrated in Fig.~\ref{fig:find-plane}. Specifically, let  $d_1$ and $d_2$ denote the distances between $(b,e)$ and the leaf node's NW and SE vertices, respectively. If $d_1 < d_2$, then we associate $(b,e)$ with the NW triangle; otherwise, we associate it with the SE triangle.

\begin{figure}[!htb]
	\centering
    \begin{subfigure}{.45\textwidth}
    	\centering
    	\includegraphics[width=1\textwidth,trim={2cm 6cm 9cm 2cm},clip]{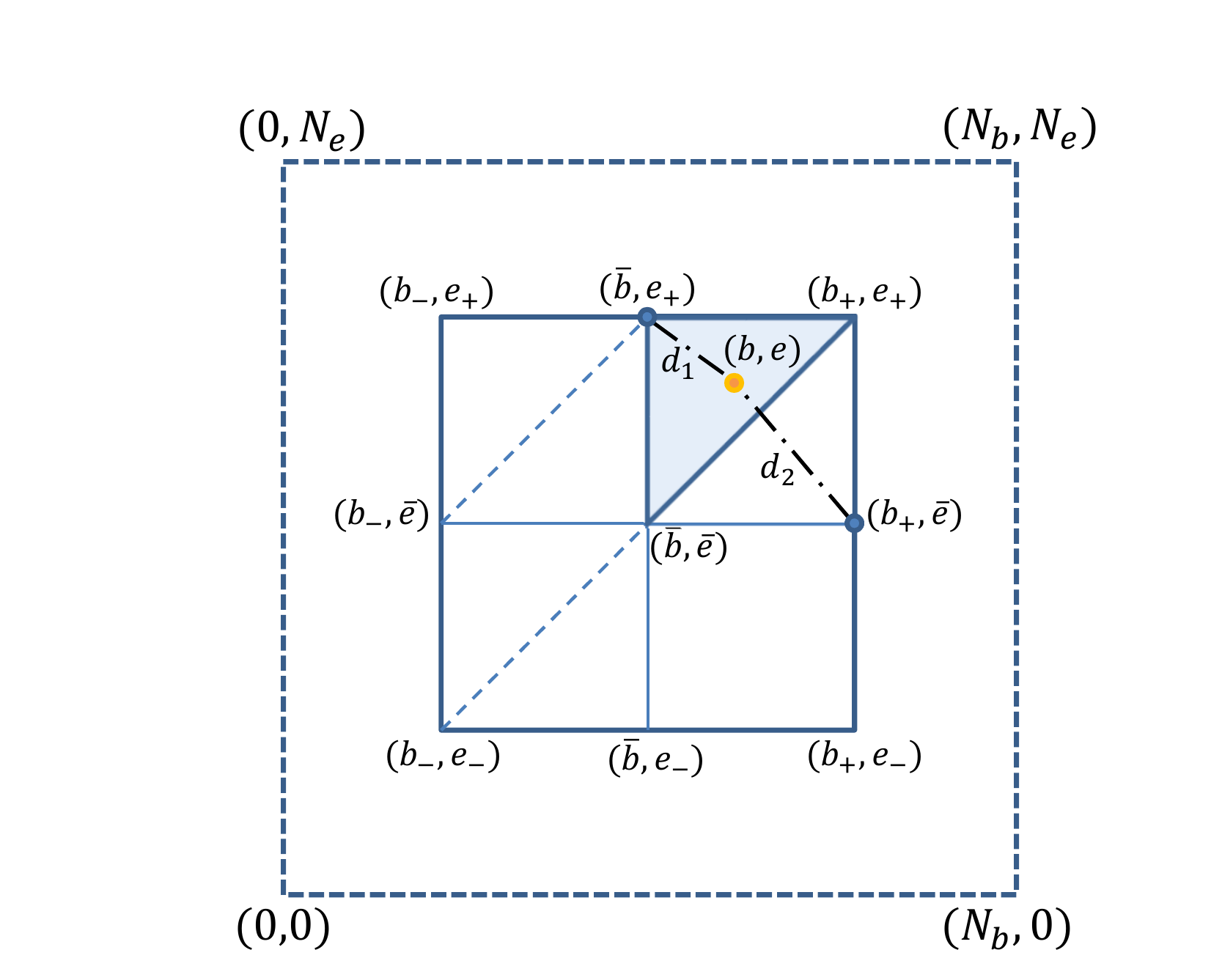}
    	\subcaption{$(b,e)$ lies within the quadtree's bounding box.}
        \label{fig:in-bound-plane}
  	\end{subfigure}
	\begin{subfigure}{.45\textwidth}
    	\centering
    	\includegraphics[width=1\textwidth,trim={2cm 6cm 9cm 2cm},clip]{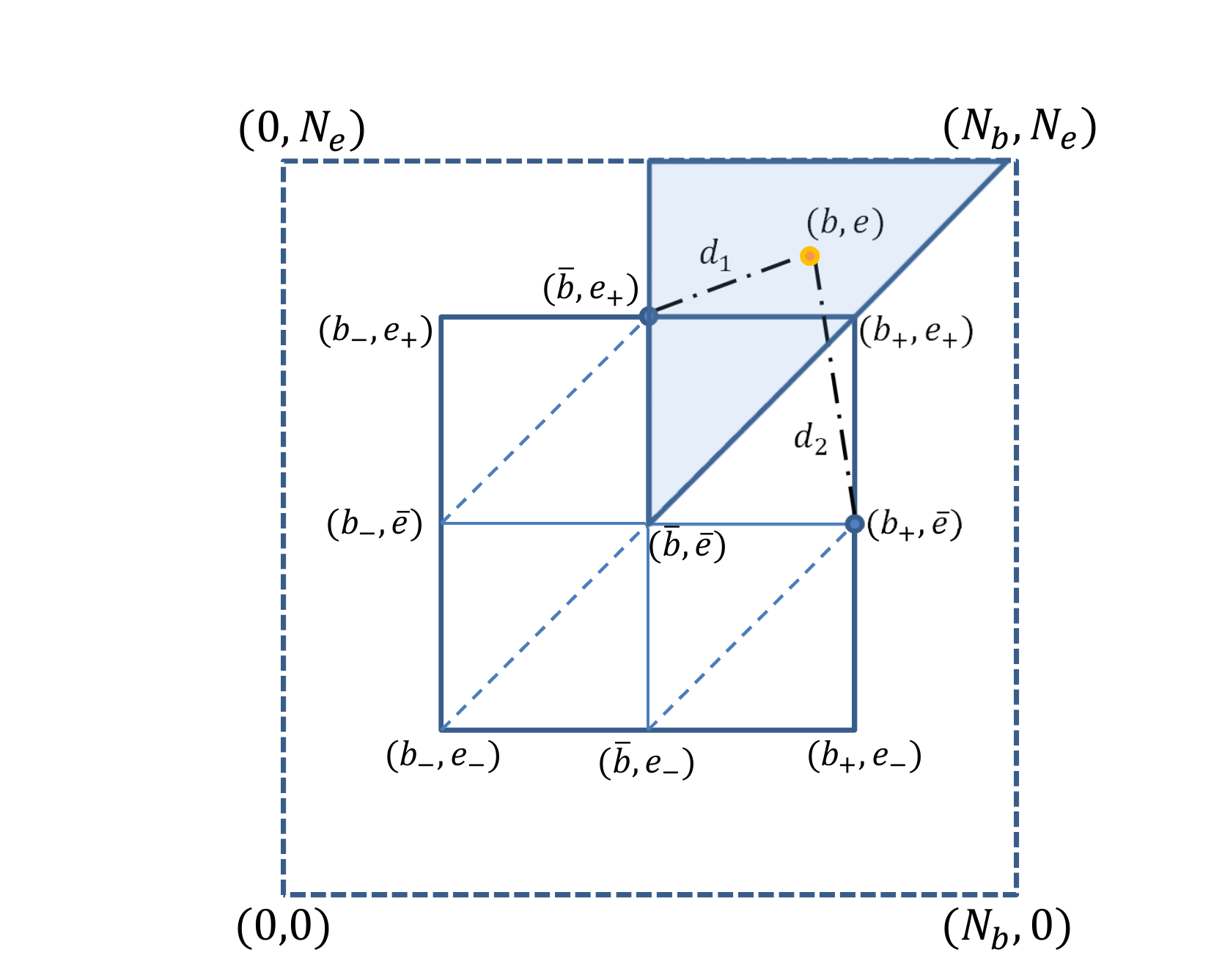}
    	\subcaption{$(b,e)$ lies outside of the quadtree's bounding box.}
        \label{fig:out-bound-plane}
  	\end{subfigure}
\caption{Associating the buffer-battery state pair $(b,e)$ with the leaf node's closest triangle.  If $d_1 < d_2$, then we use the NW triangle; otherwise, we use the SE triangle.}
\label{fig:find-plane}
\end{figure}

Denote the vertices of the selected triangle by $\mathbf{x}_i = (b_i, e_i, \pds{V}(b_i,e_i,h))$, for $i = 1,2,3$, as illustrated in Fig~\ref{fig:approx-PDSV}. These three points define a plane with normal vector $\mathbf{n} = (n_1,n_2,n_3) =(\mathbf{x}_1 - \mathbf{x}_2) \times (\mathbf{x}_1 - \mathbf{x}_3)$, where $\times$ denotes the cross product. The equation of the approximating plane can therefore be written as:
\begin{equation*}
n_1 (\pds{b} - b_1) + n_2 (\pds{e} - e_1) + n_3 ( V - \pds{V}(b_1,e_1,h)) = 0.
\end{equation*}
Finally, substituting $(b,e)$ for $(\pds{b},\pds{e})$ and solving for $V$ we get:
\begin{equation} \label{eq:calc-approx-PDSV}
V=\hat{V}(b,e,h) = \pds{V}(b_1,e_1,h) - \frac{n_1 (b - b_1) + n_2 (e - e_1)}{n_3}.
\end{equation}
Pseudocode for the function \texttt{approximate\_PDSV} is given in Algorithm~\ref{alg:approximate-PDSV}.

The following proposition shows that the maximum error resulting from a piece-wise planar approximation of the optimal PDS value function is bounded.

\begin{proposition}\label{prop:bound}
Let $\pds{V}^*$ denote the optimal PDS value function that satisfies the Bellman equation~\eqref{eq:V_to_PDSV}. Let $\hat{V}$ denote the approximate PDS value function~\eqref{eq:approx-PDSV}. Let $(b,e)$ be associated with the triangle with vertices $\mathbf{x}_i = (b_i, e_i, \pds{V}^*(b_i,e_i,h))$, for $i = 1,2,3$. It follows that
\begin{equation} \label{eq:error}
\hat{V}(b,e,h) - \pds{V}^{*}(b,e,h) \leq \max_{i \in \{1,2,3\}} \pds{V}^{*}(b_i,e_i,h) - \min_{i \in \{1,2,3\}} \pds{V}^{*}(b_i,e_i,h).
\end{equation}
\end{proposition}
\begin{proof}
The result follows from Propositions~\ref{prop:structure-PDSV-b} and~\ref{prop:structure-PDSV-e}. In particular, since $\pds{V}^{*}$ has increasing differences in $\pds{b}$ and $\pds{e}$, the plane defined by the approximating triangle provides an upper bound on the true value function. Additionally, since $\pds{V}^{*}$ and $\hat{V}$ are non-decreasing in $\pds{b}$ and non-increasing in $\pds{e}$, they are both bounded by $\min_{i \in \{1,2,3\}} \pds{V}^{*}(b_i,e_i)$ and $\max_{i \in \{1,2,3\}} \pds{V}^{*}(b_i,e_i)$ for all $(b,e)$ that lie in the approximating triangle. The result in~\eqref{eq:error} immediately follows.
\end{proof}

\begin{figure}[!htb]
\centering
  \includegraphics[width=3.45in,trim={6.1cm 6.8cm 5.8cm 2cm},clip]{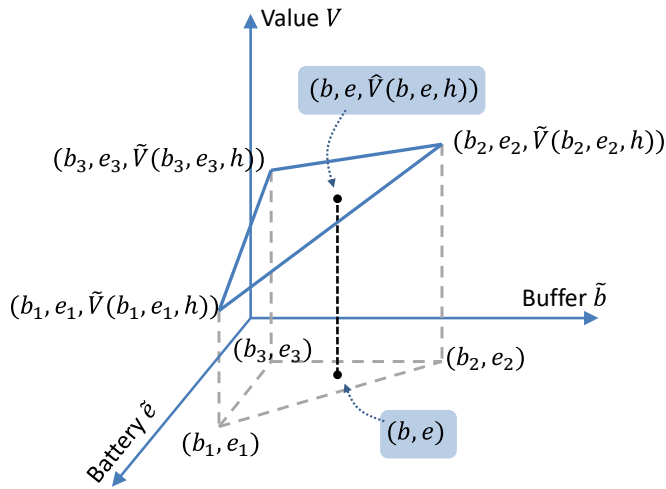}
  \caption{Calculating the approximate value $\hat{V}(b,e,h)$ of a buffer-battery state pair $(b,e) \notin \mathcal{T}(h)$ using a piece-wise planar approximation. $\hat{V}(b,e,h)$ is calculated as in~\eqref{eq:calc-approx-PDSV}.}
  \label{fig:approx-PDSV}
  \vspace{-0.5cm}
\end{figure}


\subsubsection{Dynamic grid update} \label{sec:grid-update}

The function \texttt{update\_grid} adaptively refines the piecewise-planar approximation until a predetermined maximum error threshold, $\delta$, is met. The algorithm finds the  error $\delta_\ell$ among all leaf nodes $\mathcal{T}_\ell \in \texttt{leaves}(\mathcal{T})$, where $\delta_\ell$ is calculated as the error defined on the right-hand side of \eqref{eq:error}. Subsequently, if $\max_{\ell} \delta_\ell > \delta$, then $\mathcal{T}_\ell$ is subdivided as described in Section~\ref{sec:quadtree}. Pseudocode for the function \texttt{update\_grid} is given in Algorithm~\ref{alg:DynamicGrid}.

\begin{algorithm} [!htb]
\caption{Grid Learning (update period $T=1$)}
\label{alg:grid-learning}
\begin{algorithmic}[1]
\State \textbf{initialize} $\mathcal{T}^{0}(\pds{h})$ for all $\pds{h} \in \mathcal{S}_h$ as in \eqref{eq:bounding-box}, $\pds{V}^0(\pds{b}, \pds{e}, \pds{h}) = 0$ for all $(\pds{b}, \pds{e}) \in \mathcal{T}^{0}(\pds{h})$ and $\pds{h} \in \mathcal{S}_h$, and $\delta$ to the desired error threshold
\For {time slot $n = 0, 1, 2, \ldots$}
\State Take the greedy action:
\begin{equation}
\label{eq:grid-greedy}
a^n = \argmin_{a \in \mathcal{A}(b^n, e^n)} \left\{b^n + \sum\nolimits_{f = 0}^{a} P^{f}(f | a, h^n)\hat{V}^n(b^n - f, e^n - a \cdot e_{TX}, h^n)\right\}
\end{equation}

\State Observe data arrivals $l^n$, energy arrivals $e_H^n$, and next channel state $h^{n + 1}$

\For {all $(\pds{b}, \pds{e}) \in \mathcal{T}^{n}(h^n)$}

\State $\pds{V}^{n + 1}(\pds{b}, \pds{e}, h^n) \leftarrow $ \texttt{update\_PDSV}$\bigl(\hat{V}^{n}, [\pds{b}, \pds{e}, h^n], [l^n, e_H^n, h^{n + 1}], \beta^n\bigr)$
\Comment{Algorithm~\ref{alg:update-PDSV}}

\EndFor

\State $\mathcal{T}^{n + 1}(h^n) \leftarrow $ \texttt{update\_grid}$\bigl(\pds{V}^{n}, \mathcal{T}^{n}(h^n), \delta \bigr)$
\Comment{Algorithm~\ref{alg:DynamicGrid}}

\EndFor
\end{algorithmic}
\end{algorithm}


\begin{algorithm}[!htb]
\caption{Approximate the PDS value function (\texttt{approximate\_PDSV})} \label{alg:approximate-PDSV}
\begin{algorithmic}[1]
\State \textbf{input} $\pds{V}$, $\mathcal{T}$, and $(b,e)$
\State Associate $(b,e)$ with $\mathcal{T}$'s nearest leaf node using a recursive search
\State Further associate $(b,e)$ with the leaf node's closest triangle as in Fig.~\ref{fig:find-plane}
	\State Calculate $\hat{V}(b,e,h)$ from $\pds{V}$ as in~\eqref{eq:calc-approx-PDSV}
	\State \textbf{return} $\hat{V}(b,e,h)$ 

\end{algorithmic}
\end{algorithm}


\begin{algorithm}[!htb]
\caption{Dynamic Grid Update (\texttt{update\_grid})}
\label{alg:DynamicGrid}
\begin{algorithmic}[1]
\State \textbf{input} $\pds{V}$, $\mathcal{T}$, and $\delta$
\For {each leaf $\mathcal{T}_{\ell} \in \texttt{leaves}(\mathcal{T})$}
\State Calculate the approximation error $\delta_{\ell}$ as in \eqref{eq:error}
\EndFor
\State $\delta_{\max} \leftarrow \max_{\ell} \delta_{\ell}$ and $\ell_{\max} \leftarrow \argmax_{\ell} \delta_{\ell}$
\If {$\delta_{\max} > \delta$}
\State Subdivide quadtree $\mathcal{T}_{\ell_{\max}}$ and add to $\mathcal{T}$
\EndIf
\State \textbf{return} $\mathcal{T}$
\end{algorithmic}
\end{algorithm}


\vspace{-6pt}
\subsection{Complexity Analysis}
Table \ref{tab:complexity-analysis} compares the action selection, learning update, and grid update complexities of the proposed \textit{Grid Learning} (Algorithm \ref{alg:grid-learning}) algorithm against the state-of-the-art \textit{PDS Learning} (Algorithm \ref{alg:pds-learning}) and \textit{Virtual Experience} (Algorithm \ref{alg:ve-learning}) algorithms. Note that the grid complexity is not defined for the \textit{PDS learning} and \textit{Virtual Experience} learning algorithms, as they do not include a grid update step. In the subsequent discussion, let $|\mathcal{S}|$ and $|\mathcal{A}|$ denote the set of states and actions respectively; let $|\mathcal{S}_b|$, $|\mathcal{S}_e|$ and $|\mathcal{S}_h|$ denote the number of buffer, energy, and channel states, respectively; and let $|\mathcal{L}|$, $|\mathcal{E}|$ and $|\mathcal{F}|$ denote the size of supports for the data packet arrival, energy packet arrival, and goodput distributions, respectively.

The action selection complexity of PDS learning is $\mathcal{O}(|\mathcal{F}||\mathcal{A}|)$ as, from \eqref{eq:pds-learning-greedy}, it needs to iterate over the goodput to calculate the value and also over all possible actions to find the best action. 
The learning update complexity as calculated from \eqref{eq:pds-evaluate-V} is also $\mathcal{O}(|\mathcal{F}||\mathcal{A}|)$ for similar reasons.

For the Virtual Experience algorithm described in Algorithm \ref{alg:ve-learning}, the action-selection complexity is the same as that of PDS learning, i.e., $\mathcal{O}(|\mathcal{F}||\mathcal{A}|)$. To compute the learning update complexity, we introduce a new notation, $|\Pi| = |\mathcal{S}_b \times \mathcal{S}_e|$, which denotes the total number of buffer-battery state pairs. Since virtual experience learning proceeds similar to PDS learning, but updates all buffer-battery pairs in each iteration, the per-step learning update complexity of the virtual experience algorithm evaluates to be $\mathcal{O}(|\Pi| |\mathcal{F}| |\mathcal{A}|)$.

The proposed \textit{Grid learning} algorithm features similar complexity to the Virtual Experience learning algorithm, save for the differences mentioned in Section \ref{sec:appoximation}. Thus, the space of points directly evaluated is reduced to the quadtree, $\mathcal{T}$. Additionally, Algorithm \ref{alg:approximate-PDSV} introduces a worst-case complexity of $\mathcal{O}(k)$ to determine the approximate value of a $(\pds{b}, \pds{e})$ pair in a quadtree with maximum depth $k$. Thus, the per-iteration complexity of the \textit{grid learning} algorithm is $\mathcal{O}(k|\mathcal{T}| |\mathcal{F}| |A|$), and the additional complexity  per call of the \textit{update\_grid} method is $ k |\mathcal{T}|$ to check if the quadtree needs to be subdivided further.


\begin{table*}
\centering
\caption{Action-Selection and Iteration complexity of several learning algorithms} 
\label{tab:complexity-analysis}
\begin{tabu}  to 1.0\textwidth { | X[c] | X[c] | X[c] | X[c] |}
\hline
\textbf{Algorithm} & \textbf{Action Selection Complexity} & \textbf{Iteration Complexity} & \textbf{Grid Update Complexity}\\
\hline 
PDS Learning & $\mathcal{O}(|\mathcal{F}| |\mathcal{A}|)$ & $\mathcal{O}(|\mathcal{F}||\mathcal{A}|)$ & -\\ 
\hline
Virtual Experience Learning & $\mathcal{O}(|\mathcal{F}||\mathcal{A}|)$ & $\mathcal{O}(|\Pi||\mathcal{F}||\mathcal{A}|)$ & - \\
\hline
Grid Learning & $\mathcal{O}(k|\mathcal{F}| |\mathcal{A}|)$ & $\mathcal{O}$($k|\mathcal{T}| |\mathcal{F}| |\mathcal{A}|)$ & $k|\mathcal{T}|$ \\
\hline 
\end{tabu}
\vspace{-0.2cm}
\end{table*}

\section{Simulation Results}\label{sec:sim}

We now present our simulation results. In Section~\ref{subsec:sim-setup}, we describe the simulation setup. In Section~\ref{subsec:learning-algo-comparison}, we compare the proposed grid learning algorithm against Q-learning, PDS learning, virtual experience learning, and the optimal policy. Finally, in Section~\ref{subsec:learning-approx-results}, we explore how the approximation error threshold affects learning performance and study the behavior of our adaptive grid refinement algorithm. 

\vspace{-6pt}
\subsection{Simulation Setup}\label{subsec:sim-setup}

The simulation parameters used in our MATLAB-based simulator are described in Table \ref{tab:simulation-parameters}. We assume that the buffer and battery have sizes $N_b = 32$ data packets and $N_e = 32$ energy packets, respectively, and that there are $N_h = 8$ channel states with PLRs $q(h) = 0.1, 0.2, \ldots, 0.8$. This yields a large state space comprising a total of $(N_b+1) \times (N_e+1) \times N_h = 8712$ states. We assume that the channel fading state is known to the transmitter at the beginning of each time slot; however, the Markovian channel transition probability function, $P^h(\cdot|h)$, is unknown a priori. We further assume that the data and energy packet arrival distributions, $P^l(\cdot)$ and $P^{e_H}(\cdot)$, respectively, are Bernoulli, but are unknown a priori. Finally, we set the discount factor $\gamma = 0.98$ to balance present and expected future costs and to optimize the long term behavior of the scheduling policy.

\begin{table}[!ht]
\centering
\caption{Simulation parameters.} 
\label{tab:simulation-parameters}
\begin{tabu}  to  \textwidth { | X[c] | X[c] | X[c] | X[c]| }
\hline
Parameter & Value & Parameter & Value \\
\hline \hline
Packet Buffer Size, $N_b$ & 32 & Transmit Action, $a$ & $\left\lbrace 0, 1 \right\rbrace$ \\ 
\hline
Battery Size, $N_e$ & 32 & Packet Transmit Energy, $e_{\text{TX}}$ & 1 \\
\hline
Channel States $h \in \mathcal{H}$ & $\left\lbrace 1, 2, \ldots, 7, 8 \right\rbrace$ & Discount Factor, $\gamma$ & 0.98 \\
\hline
Error Rate, $q(h)$ & $\left\lbrace 0.8, 0.7, \ldots , 0.1 \right\rbrace$ & Simulation Duration (slots) & 50,000 \\ 
\hline
Packet Arrivals (pkts/slot) & $\left\lbrace 0, 1 \right\rbrace$ & VE Update period, $T_{\text{VE}}$ & 10
\\
\hline
Energy Arrivals (pkts/slot) & $\left\lbrace 0, 1 \right\rbrace$ & Grid Update Period, $T_{\text{grid}}$ & $\left\lbrace 10, 50, 100 \right\rbrace$ \\
\hline
Data Packet Arrival Distribution, $P^l(l)$ & $\text{Bern}(p), p \in \left\lbrace 0.1,0.2, \ldots, 0.6 \right\rbrace$ & Approximation Error Threshold, $\delta$ & $\left\lbrace 5, 7.5, 10, \ldots, 45 \right\rbrace$ \\
\hline
Energy Packet Arrival Distribution, $P^{e_H}(e_H)$ & $\text{Bern}(0.7)$ & Packet Overflow Penalty, $\eta$ & 50 \\
\hline
\end{tabu}
\end{table}





\vspace{-6pt}
\subsection{Learning Algorithm Comparison}\label{subsec:learning-algo-comparison}

We implement the Q-learning algorithm as described in~\cite{sutton1998reinforcement}, PDS learning algorithm as described in Section \ref{sec:pds-learning} and Algorithm~\ref{alg:pds-learning}, the VE learning algorithm as described in Section~\ref{sec:ve-learning} and Algorithm~\ref{alg:ve-learning}, and the grid learning algorithm as described in Section~\ref{sec:2d-grid-update} and Algorithm~\ref{alg:grid-learning}. Simulation results using the parameters summarized in Table \ref{tab:simulation-parameters} are presented in Fig. \ref{fig:algo-comparison} for numerous simulations with duration 50,000 time slots, data packet arrival distribution $P^l(l) \sim \text{Bern}(0.4)$, energy packet arrival distribution $P^{e_H}(e_H) \sim \text{Bern}(0.7)$, error tolerance $\delta = 10$, and initial states $b^0 = e^0 = 0$. 



\begin{figure}[!htb]
	\centering
    
    \begin{subfigure}{0.48\textwidth}
    	\centering
        \vspace{-3.5cm}
    	\includegraphics[width=3.2in,trim={1.4cm 1cm 1cm 1cm},clip]{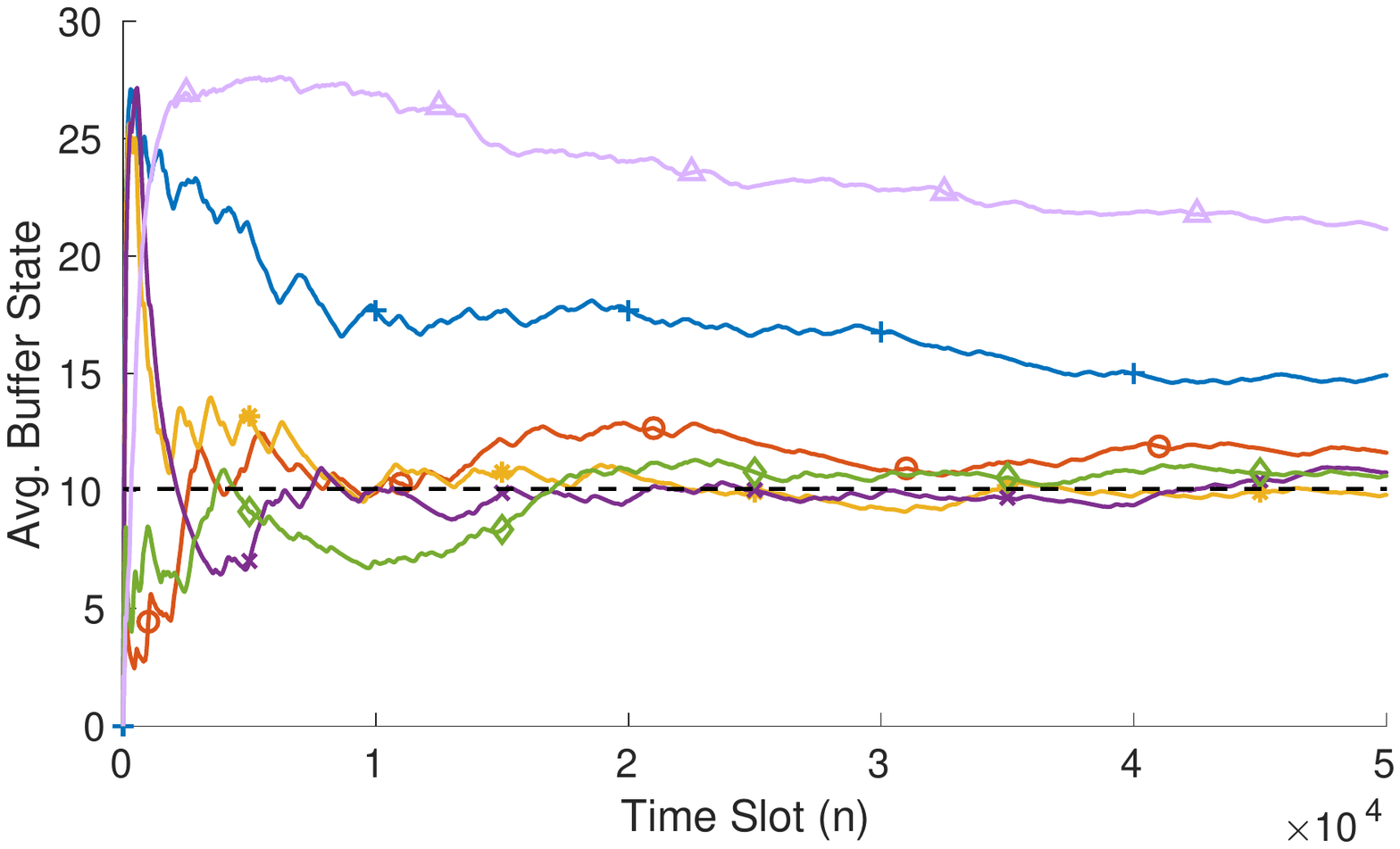}
       \vspace{-4cm}
    	\subcaption{Average Buffer Occupancy vs. Time}
        \label{fig:algo-comparison-buffer}
  	\end{subfigure}
    \begin{subfigure}{0.48\textwidth}
    	\centering
        \vspace{-3.5cm}
    	\includegraphics[width=3.2in,trim={1.4cm 1cm 1cm 1cm},clip]{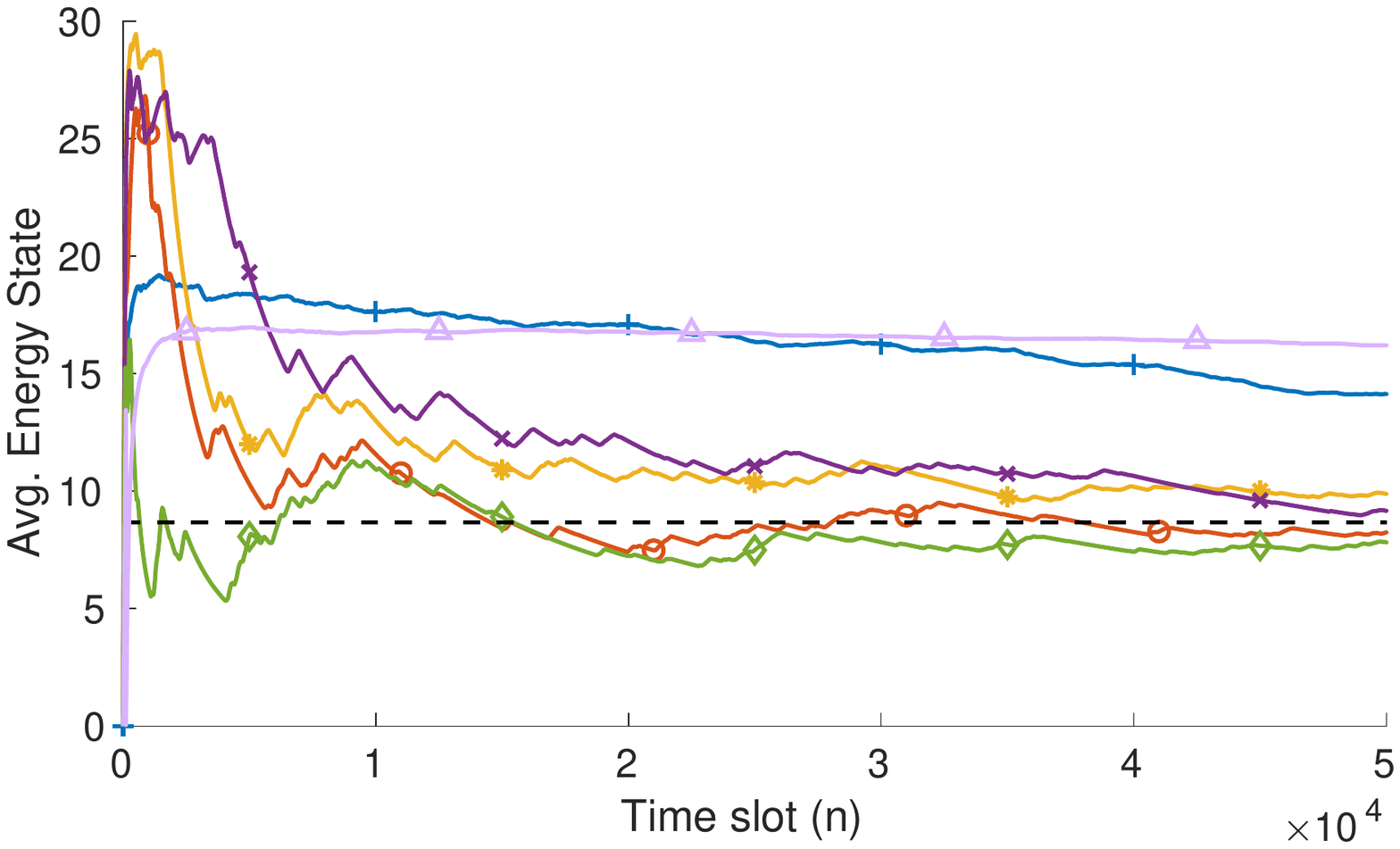}
        \vspace{-4cm}
    	\subcaption{Average Battery Occupancy vs. Time}
        \label{fig:algo-comparison-battery}
  	\end{subfigure}
    
    \begin{subfigure}{0.48\textwidth}
    	\centering
        \vspace{-3cm}
    	\includegraphics[width=3.2in,trim={1.4cm 1cm 1cm 1cm},clip]{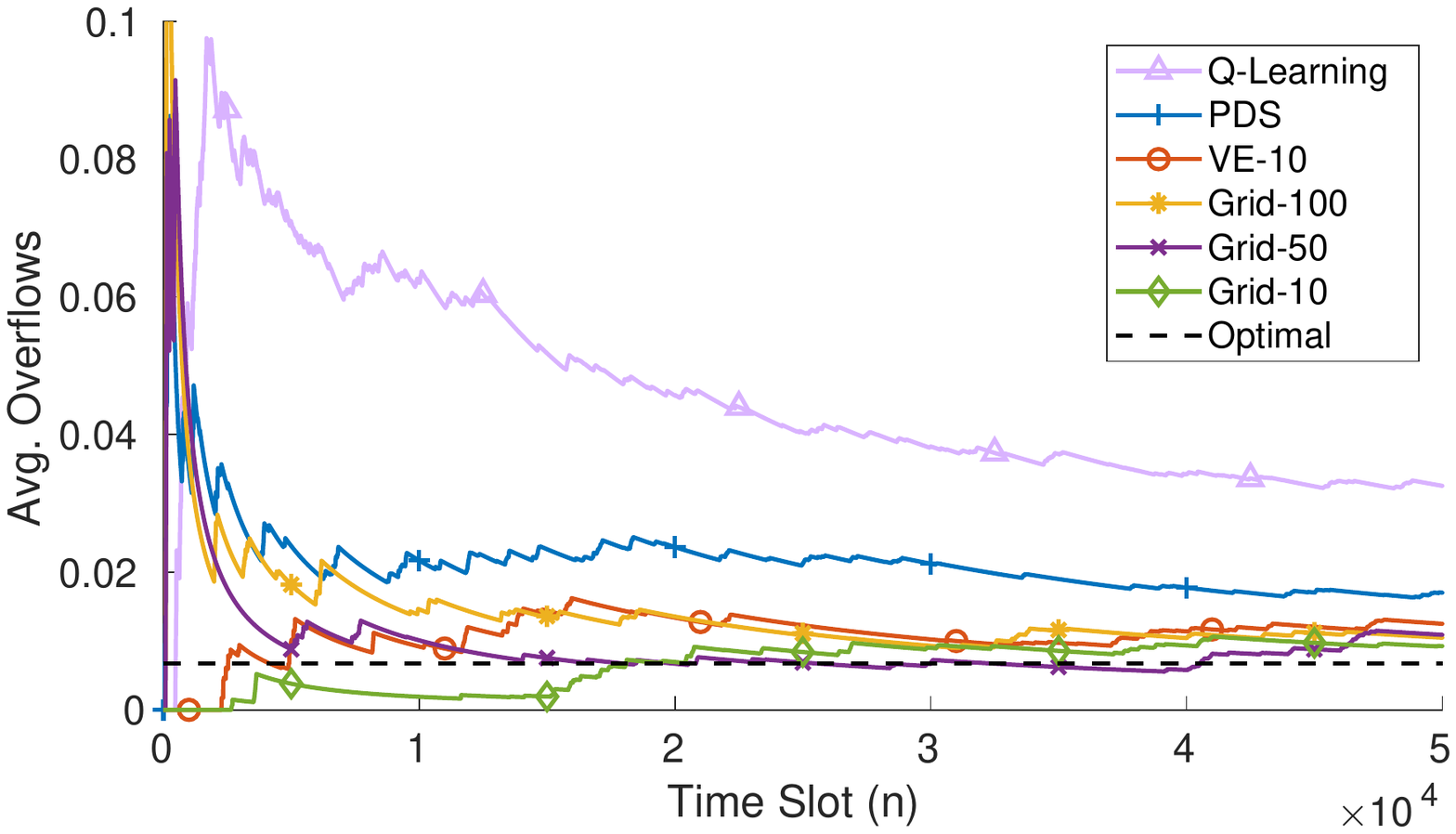}
        \vspace{-4cm}
    	\subcaption{Average Overflows vs. Time}
        \label{fig:algo-comparison-overflows}
  	\end{subfigure}
\caption{Performance comparison of the grid, PDS, and virtual experience learning algorithms.}
\label{fig:algo-comparison}
\vspace{-0.5cm}
\end{figure}

In Fig. \ref{fig:algo-comparison}, the curves labeled ``Grid-$T$'' are obtained using the grid learning algorithm with updates every $T = 10, 50, 100$ time slots; the curve labeled ``VE-10'' is obtained using the VE learning algorithm with updates every 10 time slots; and, the curves labeled ``Q-learning,'' ``PDS,'' and ``Optimal'' are obtained from the Q-learning algorithm, PDS learning algorithm, and optimal policy, respectively. Fig. \ref{fig:algo-comparison-buffer} illustrates the average buffer occupancy versus time; Fig. \ref{fig:algo-comparison-battery} illustrates the average battery occupancy versus time; and Fig. \ref{fig:algo-comparison-overflows} illustrates the average buffer overflows versus time.

The Q-learning algorithm predictably performs worse than the other algorithms. This is due due to the fact that: 1) it requires action exploration~\cite{kaelbling1996reinforcement,sutton1998reinforcement}, so it frequently chooses sub-optimal actions even if it has found the optimal action; and 2) it can only learn about one state-action pair in each time slot.   
The PDS learning algorithm, although better than Q-learning, also takes an unacceptably long time to converge to the optimal solution because it can only learn about one PDS in each time slot. 
We observe that ``Grid-10'' achieves comparable performance to both ``Optimal' and ``VE-10'' in under 20,000 time slots. Importantly, the grid learning algorithm achieves this by updating 93\% fewer states at a time compared to VE learning (at most 69 states for grid learning versus $(N_b+1) \times (N_e+1) = 1089$ for VE learning) and without any a priori knowledge about the channel, data arrival, and energy harvesting dynamics as is required to compute the optimal solution.
Owing to this, a near-optimal transmission policy can be efficiently learned online on an EHS. 
Both ``Grid-50'' and ``Grid-100'' achieve near-optimal performance that is comparable to VE learning within 50,000 time slots. Intuitively, grid learning performs better with more frequent updates.

Fig.~\ref{fig:algo-comparison} also reveals how the system evolves over time. Since the learning algorithms have no a priori knowledge of the dynamics, they operate with suboptimal policies until they gain sufficient experience through their interactions with the environment. This leads to an initial surge in the buffer occupancy, battery occupancy, and buffer overflows, as the EHS harvests energy from the environment, but has not yet learned when to transmit data packets. Q-learning and PDS learning perform particularly poorly in this ``cold start'' phase because, unlike VE and grid learning, they have to actually experience large backlogs and packet overflows to learn how to avoid them.


\begin{figure}[p]
	\centering
    \begin{subfigure}{.45\textwidth}
    	\centering
    	\includegraphics[width=0.9\textwidth,trim={1cm 7.5cm 1cm 7cm},clip]{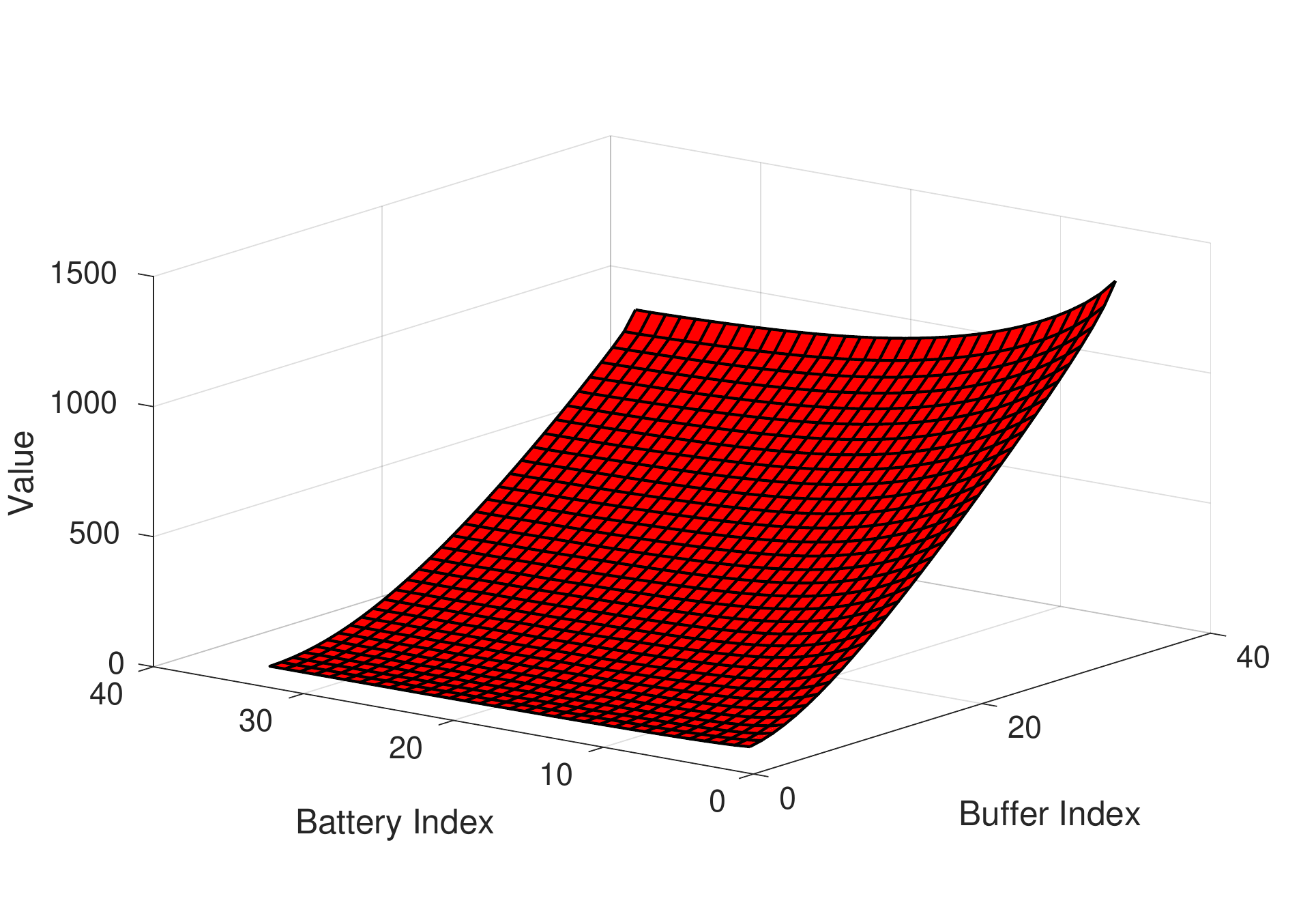}
        \vspace{-8pt}
    	\subcaption{Optimal PDS Value Function ($\delta = 0$)}
        \label{fig:optimal-value-function-surf}
  	\end{subfigure}
	\begin{subfigure}{.45\textwidth}
    	\centering
    	\includegraphics[width=0.9\textwidth,trim={1cm 7.5cm 1cm 7cm},clip]{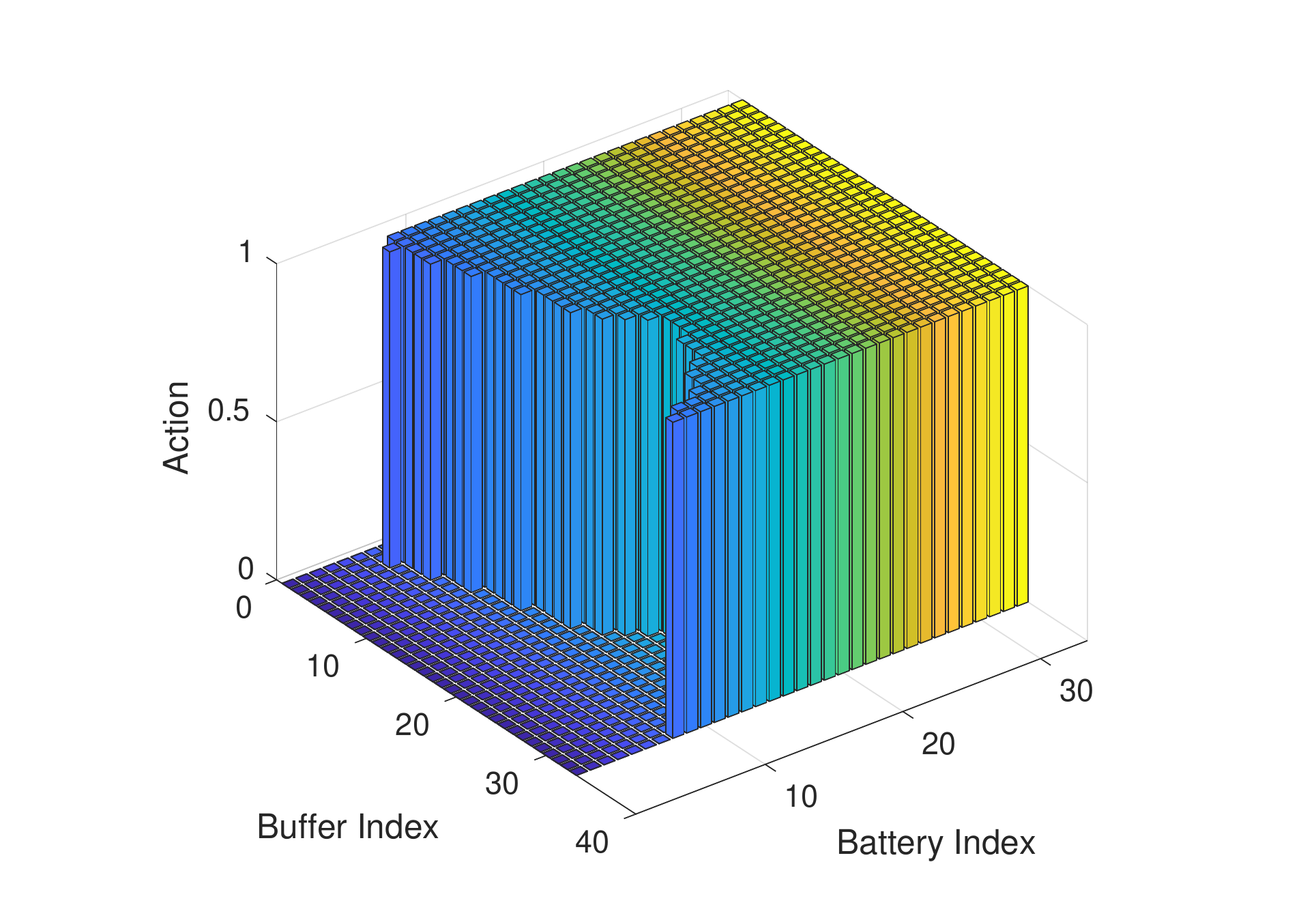}
        \vspace{-8pt}
    	\subcaption{Optimal Policy ($\delta = 0$)}
        \label{fig:optimal-policy}
  	\end{subfigure}
    
    \begin{subfigure}{.45\textwidth}
    	\centering        
    	\includegraphics[width=0.9\textwidth,trim={1cm 8cm 1cm 6cm},clip]{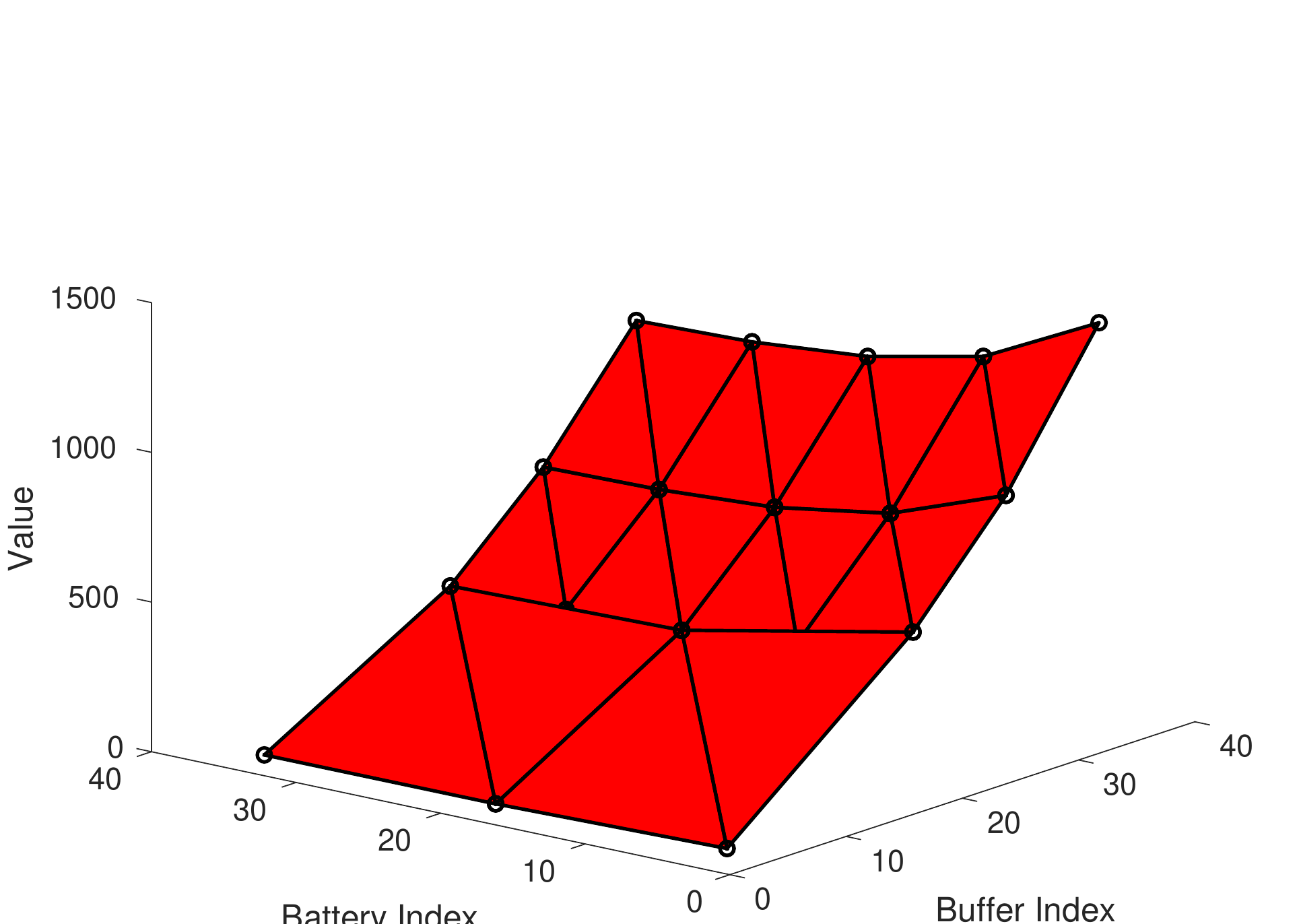}
        \vspace{-8pt}
    	\subcaption{Approximate PDS Value Function ($\delta = 10$)}
        \label{fig:piecewise-planar-value-delta-10}
  	\end{subfigure}
	\begin{subfigure}{.45\textwidth}
    	\centering
    	\includegraphics[width=0.9\textwidth,trim={1cm 7.5cm 1cm 7cm},clip]{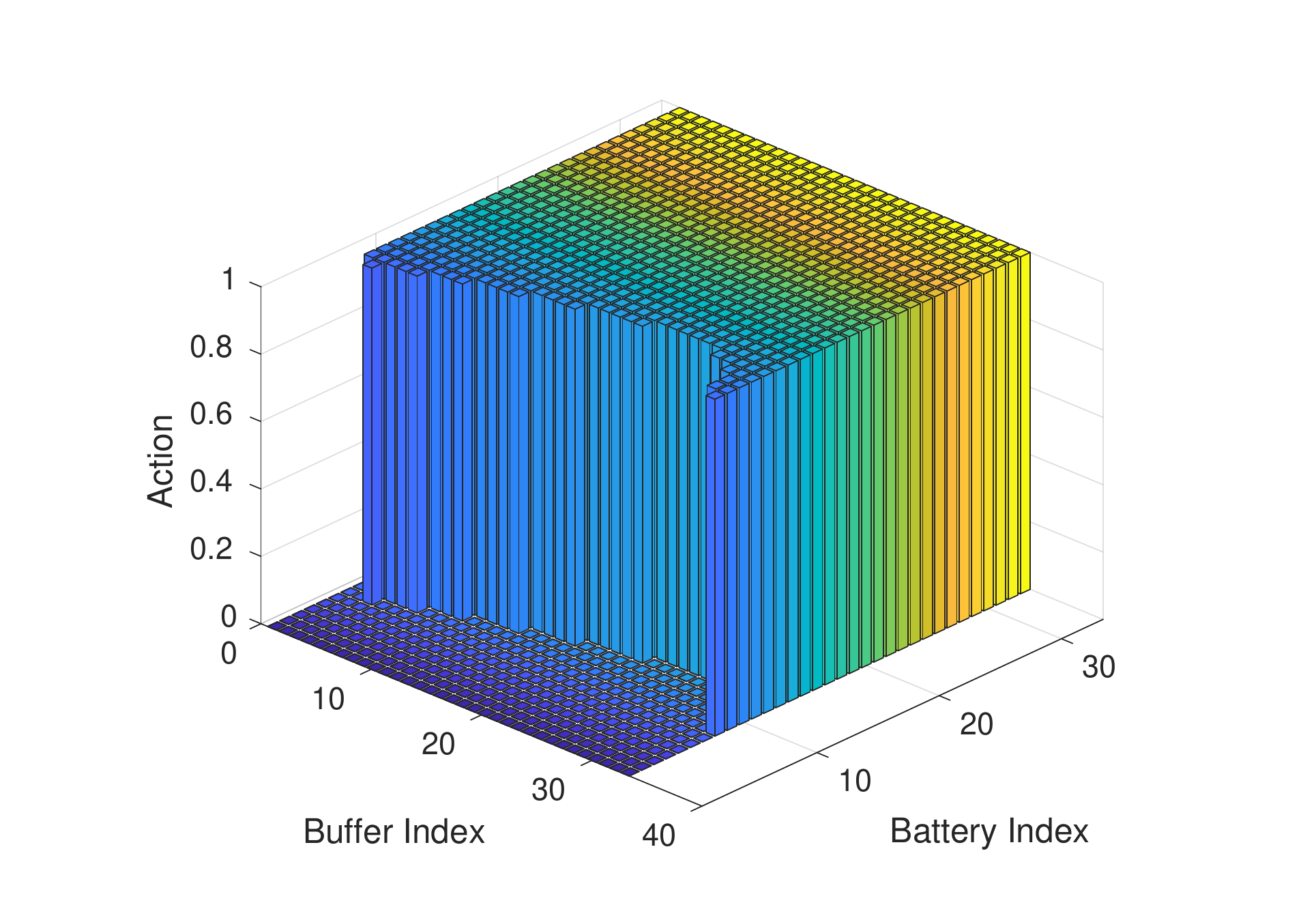}
        \vspace{-8pt}
    	\subcaption{Policy ($\delta = 10$)}
        \label{fig:policy-delta-10}
  	\end{subfigure}
    
  	\begin{subfigure}{.45\textwidth}
  		\centering
      	\includegraphics[width=0.9\linewidth, ,trim={1cm 7.5cm 1cm 7cm},clip]{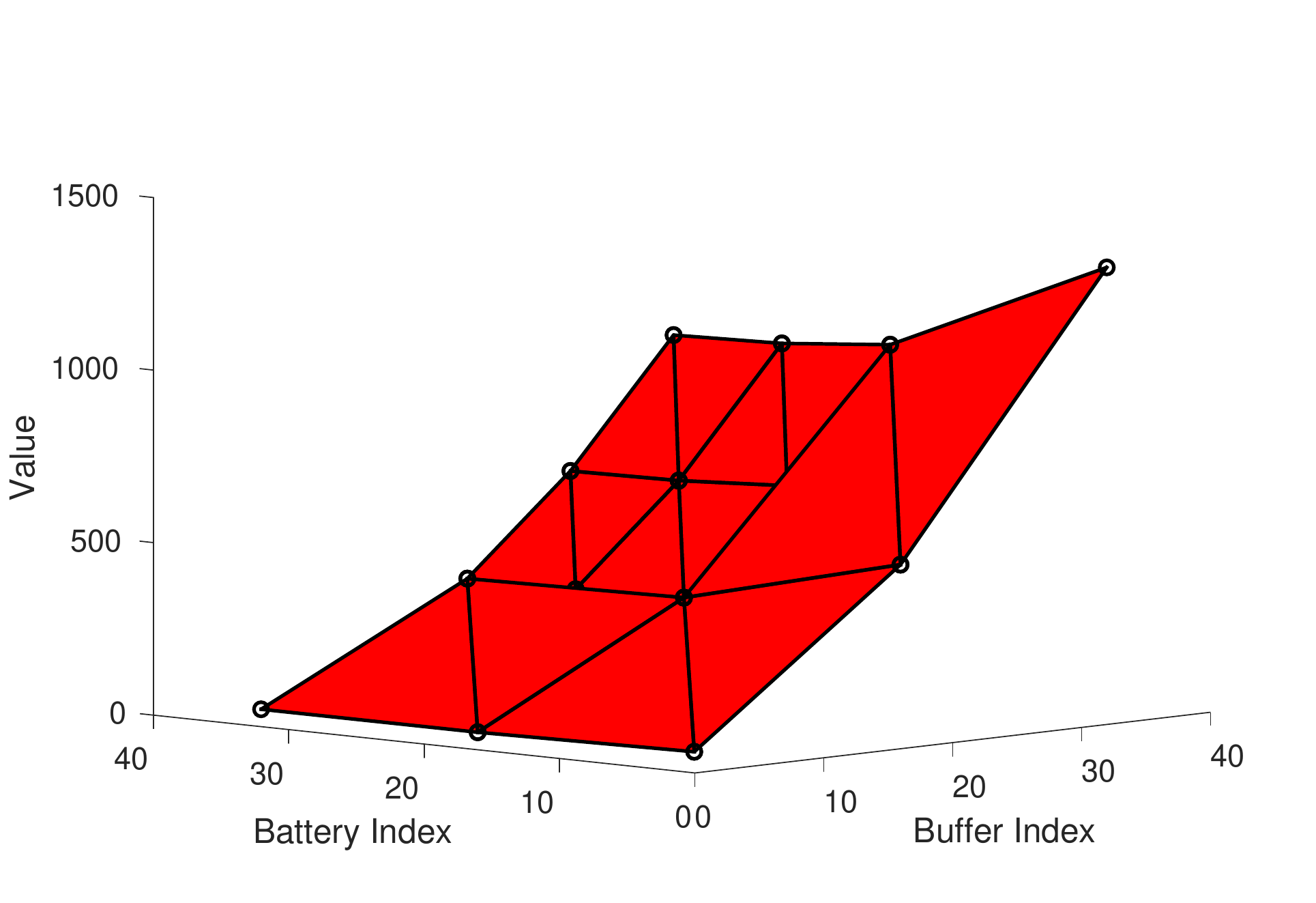}
        \vspace{-8pt}
    	\subcaption{Approximate PDS Value Function ($\delta = 20$)}\label{fig:value-function-delta-20}
  	\end{subfigure}
	\begin{subfigure}{.45\textwidth}
  		\centering
      	\includegraphics[width=0.9\linewidth, ,trim={1cm 7.5cm 1cm 7cm},clip]{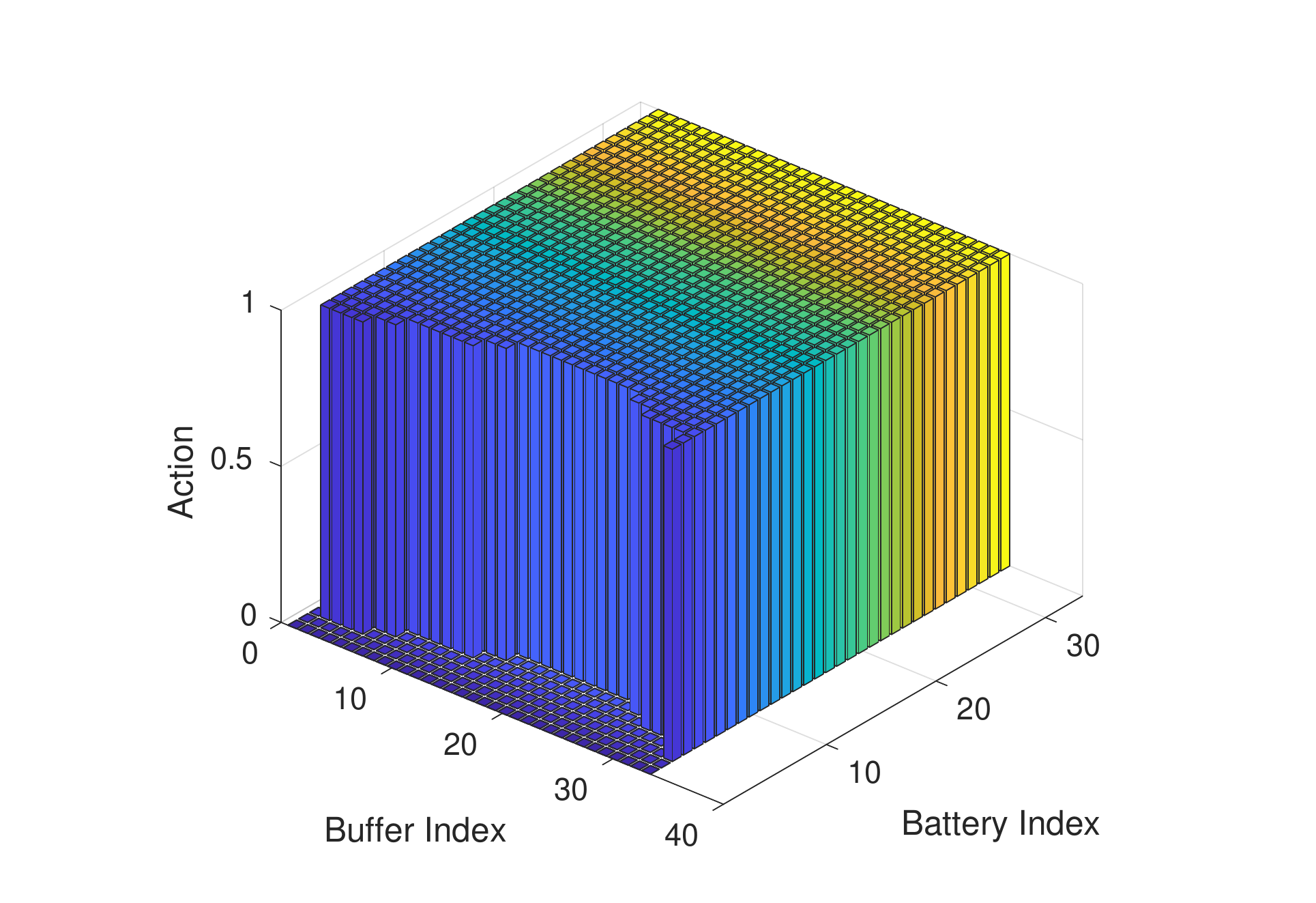}
        \vspace{-8pt}
    	\subcaption{Policy ($\delta = 20$)}\label{fig:policy-delta-20}
  	\end{subfigure}
    
    \begin{subfigure}{.45\textwidth}
  		\centering
      	\includegraphics[width=0.9\linewidth, ,trim={1cm 7.5cm 1cm 7cm},clip]{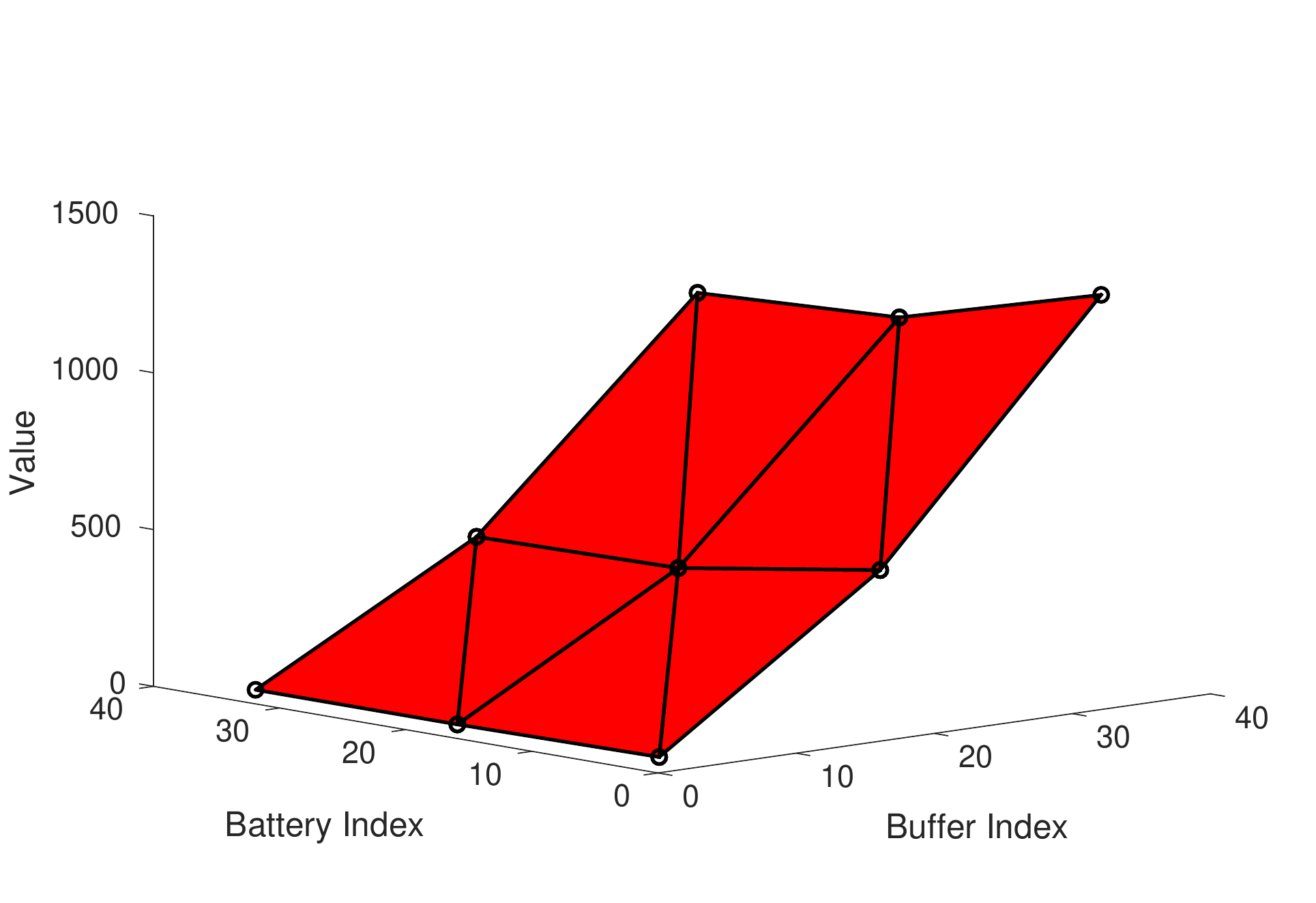}
        \vspace{-8pt}
    	\subcaption{Approximate PDS Value Function ($\delta = 30$)}\label{fig:value-function-delta-30}
  	\end{subfigure}
	\begin{subfigure}{.45\textwidth}
  		\centering
      	\includegraphics[width=0.9\linewidth, ,trim={1cm 7.5cm 1cm 7cm},clip]{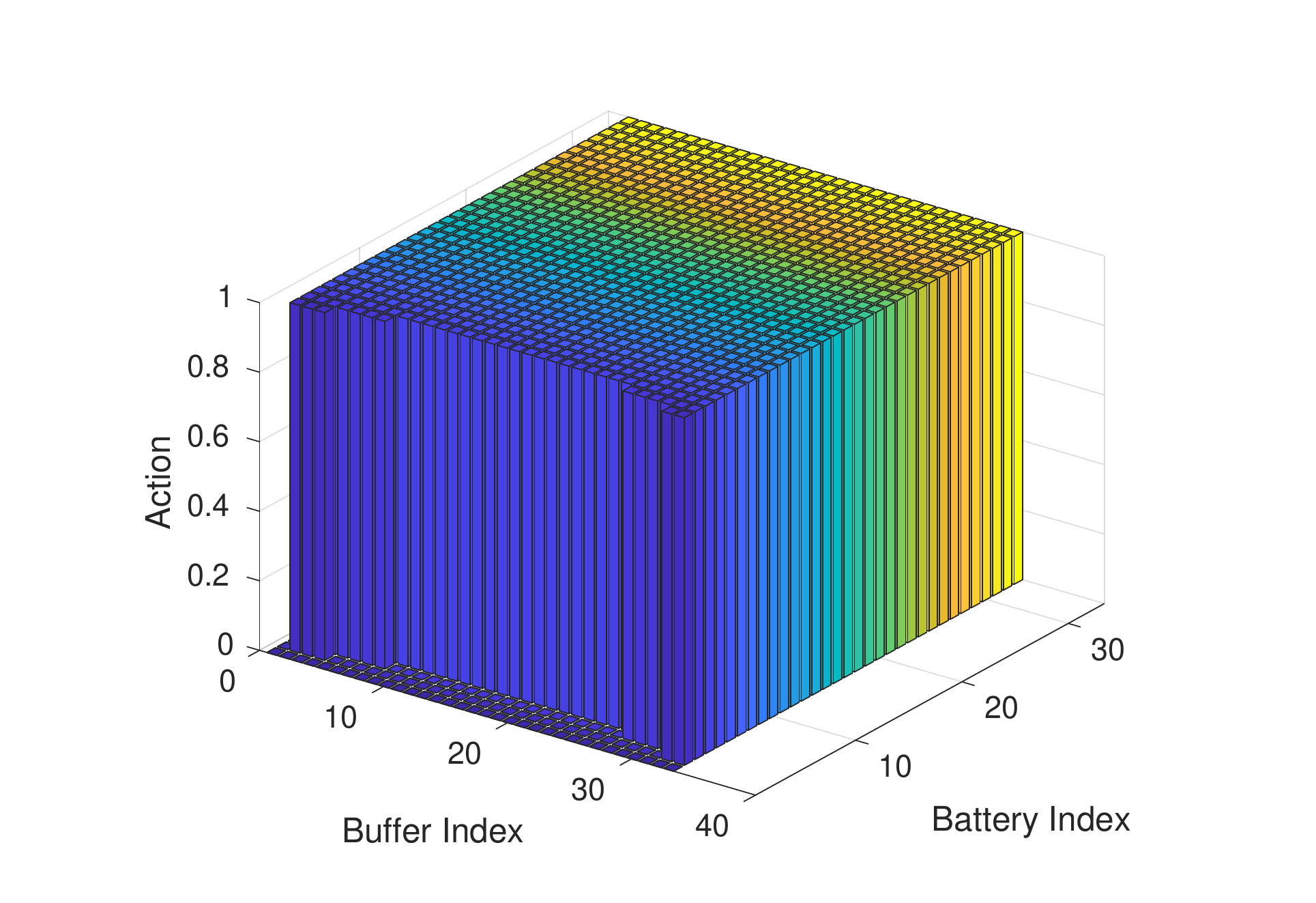}
        \vspace{-8pt}
    	\subcaption{Policy ($\delta = 30$)}\label{fig:policy-delta-30}
  	\end{subfigure}

\caption{PDS value functions and their associated policies for different error thresholds ($q(h) = 0.8$, $P^l(l) \sim \text{Bern}(0.2)$, and $P^{e_H}(e_H) \sim \text{Bern}(0.7)$).}
\label{fig:value-function-and-policy}
\end{figure}



    


\begin{figure}[!htb]
\centering
  \includegraphics[clip, trim = 1cm 7.5cm 1cm 8cm, width=3.45in]{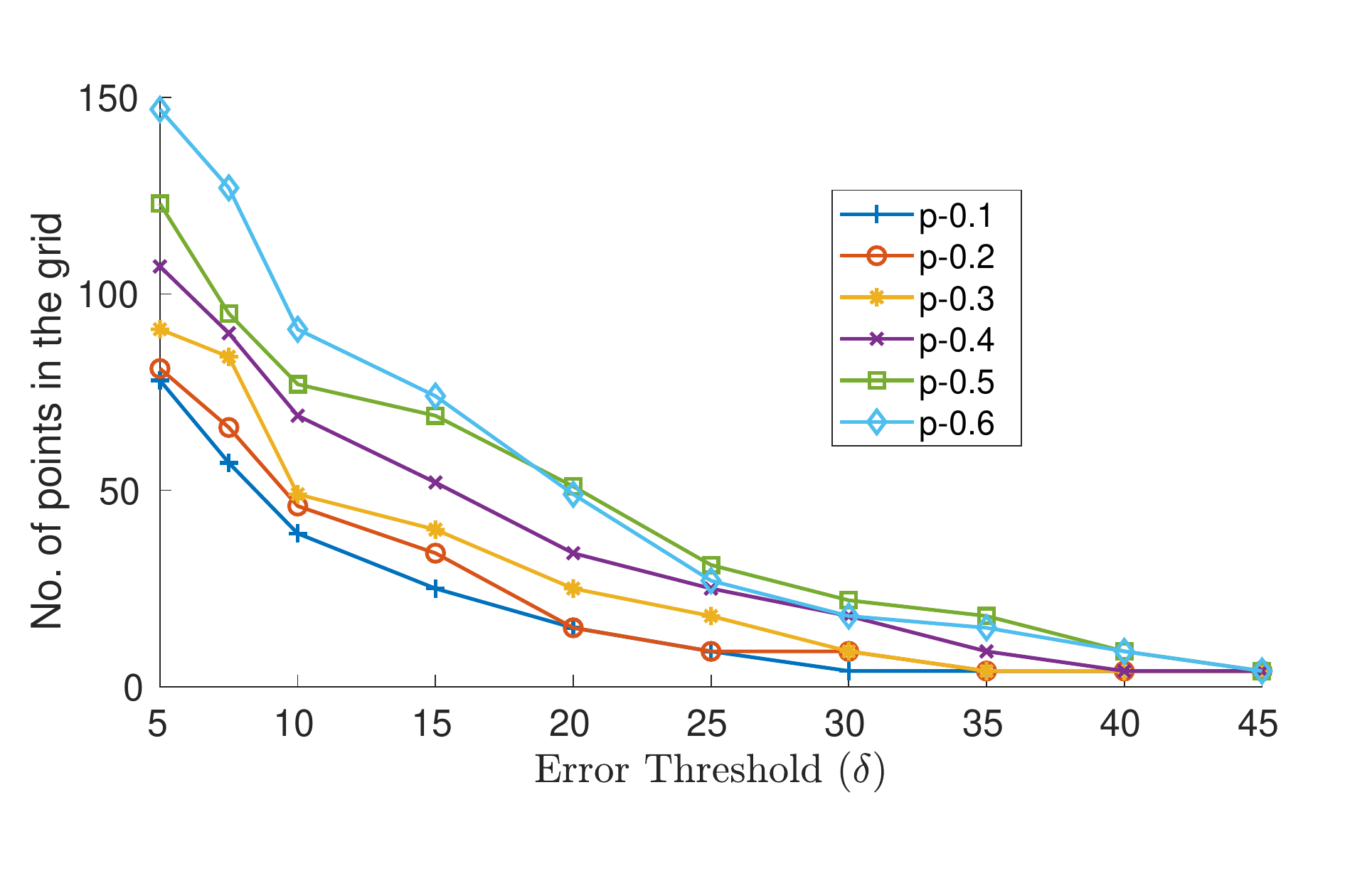}
  \vspace{-0.5cm}
  \caption{Grid Points vs. Error Threshold $\delta$ ($q(h)=0.5$, $P^l(l) \sim \text{Bern}(p)$, and $P_{e_H} \sim \text{Bern}(0.7)$).}
  \label{fig:grid-learning-pts-vs-error}
  \vspace{-0.5cm}
\end{figure}

\vspace{-6pt}
\subsection{Effect of the Approximation Error Threshold} \label{subsec:learning-approx-results}
In this section, we investigate the effect of the approximation error threshold $\delta$ on the grid learning algorithm.
All of the results in this section were taken after 50,000 time slot simulations with grid learning updates applied every $T = 100$ slots.

In Fig. \ref{fig:value-function-and-policy}, we compare several approximate PDS value functions ($\delta = 10, 20, 30$) against the optimal PDS value function ($\delta = 0$) in the worst channel state (PLR $q(h) = 0.8$) with data packet arrival distribution $P^l(l) \sim \text{Bern}(0.2)$ and energy packet arrival distribution $P^{e_H}(e_H) \sim \text{Bern}(0.7)$. We also compare their associated policies. 
In Fig. \ref{fig:optimal-value-function-surf}, we observe that the optimal PDS value function is non-decreasing and has increasing differences in the buffer state and is non-increasing and has increasing differences in the energy state (cf. Propositions~\ref{prop:structure-PDSV-b} and~\ref{prop:structure-PDSV-e}). By design, relaxing the error tolerance leads to coarser piece-wise planar approximations of the PDS value function. For instance, at approximation error thresholds 0, 10, 20, and 30, the PDS value function is represented by 1089, 18, 14, and 9 states, respectively. The actual maximum errors between the optimal and approximate PDS value functions are 8.3, 17.1 and 27.9.
Interestingly, we can also see that the policies in Fig.~\ref{fig:value-function-and-policy} become more aggressive as we relax the error threshold, i.e., they choose to transmit packets at lower and lower battery states. 


\begin{figure}[!htb]
\centering
  \includegraphics[clip, trim = 1cm 8.5cm 1cm 9cm, width=3.45in]{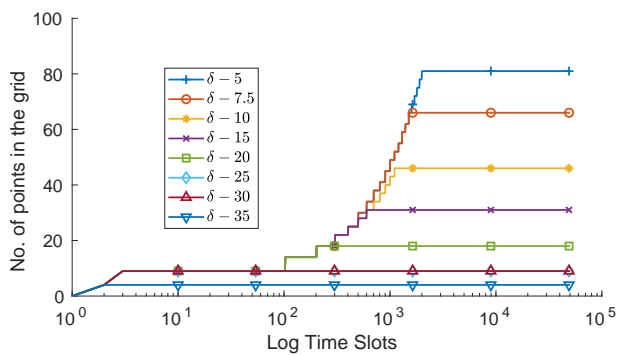}
  \caption{Grid Points vs. Time ($q(h)=0.5$, $P^l(l) \sim \text{Bern}(0.2)$, and $P_{e_H}(e_H) \sim \text{Bern}(0.7)$).}
  \label{fig:grid-pts-vs-time}
  \vspace{-0.5cm}
\end{figure}


Fig. \ref{fig:grid-learning-pts-vs-error} illustrates the number of grid points used to approximate the PDS value function versus the approximation error threshold $\delta$ for several data packet arrival rates. The measurements were taken from the approximate PDS value function in channel state $h$ with PLR $q(h) = 0.5$. 
These results further highlight that the number of grid points used in the PDS value function approximation decreases as the approximation error threshold increases. This intuitively follows from the fact that higher (resp. lower) error thresholds can be met by coarser (resp. finer) quadtree decompositions. 
We also observe that, for a fixed energy packet arrival rate, the number of grid points needed to meet a given error threshold roughly increases with the data packet arrival rate. This happens because the PDS value function's slope increases with the data packet arrival rate, which results in a larger approximation error at a fixed quadtree decomposition level (cf. Proposition~\ref{prop:bound}).
For instance, at an expected arrival rate of 0.6 packets/slot (i.e., $P^l(l) \sim \text{Bern}(0.6)$), the number of grid points needed to approximate the PDS value function within an error threshold of $\delta = 5$ is close to 150 points, which is almost twice the number of grid points needed to meet the same error threshold with an expected arrival rate of 0.1 packets/slot. This demonstrates that the grid learning algorithm can adapt to the experienced dynamics.

Fig.~\ref{fig:grid-pts-vs-time} illustrates how the quadtree decomposition evolves over time to meet different approximation error thresholds. The measurements were taken from the approximate PDS value function in channel state $h$ with PLR $q(h) = 0.5$.
As before, the terminal number of grid points is lower for higher approximation error thresholds, $\delta$.
From the figure, we can see that the grid undergoes a lot of refinement in the first 2000 time slots to meet the error threshold. 
This manifests as a step-wise increase in the number of grid points every $T_{\text{grid}} = 100$ time slots. 
Note that, subdividing a leaf node can introduce 1-5 new grid points depending on the refinement level of the surrounding leaf nodes; therefore, the step sizes are variable over time.


\begin{figure} [!htb]
\centering
  \includegraphics[clip, trim = 1cm 8.5cm 0cm 8cm, width=3.45in]{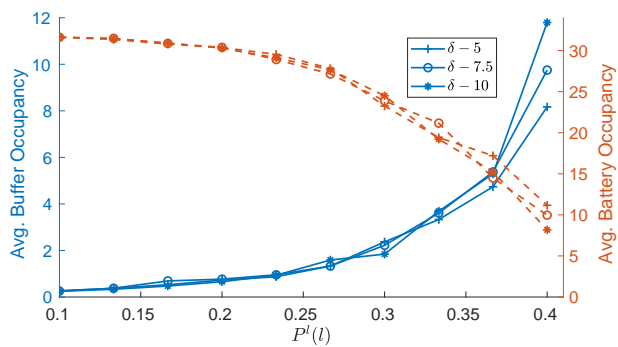} 
  \caption{Average Buffer Occupancy vs. Average Battery Occupancy ($P_{e_H}(e_H) \sim \text{Bern}(0.7)$).}
  \label{fig:grid-learning-avg-buffer-vs-energy}
  \vspace{-0.5cm}
\end{figure}

Fig. \ref{fig:grid-learning-avg-buffer-vs-energy} illustrates the average buffer and battery occupancies versus the data packet arrival rate at three different error thresholds. 
As expected, for a fixed energy packet arrival rate, the average buffer occupancy displays complementary behavior to the average battery occupancy.
This is because, at low data arrival rates, the buffer size can be kept small using a small fraction of the available energy. However, at high data arrival rates, more of the available energy is needed to keep the buffer occupancy from growing. In parallel, as the data arrival rate increases towards the channel's maximum service rate, the average queue backlog increases.  
From Fig. \ref{fig:grid-learning-avg-buffer-vs-energy}, we also observe that tighter error thresholds yield better buffer-battery (and, consequently, delay-energy) trade-offs. For instance, $\delta = 5$ results in a lower average buffer occupancy while maintaining a higher average battery occupancy than $\delta = 10$. This can be explained by the fact that more accurate PDS value function approximations translate to better transmission scheduling policies.

\section{Conclusion}
\label{sec:conclusion}
Foresighted decision making is required to optimize the performance of resource constrained communication systems. In practice, however, the lack of a priori knowledge about the system's experienced dynamics presents a major challenge.
Online RL represents a natural paradigm for overcoming this challenge, but generic RL algorithms are often unable to meet the stringent requirements of such systems in terms of memory, complexity, and convergence speed. Evidently, this requires carefully exploiting the structure of the problem at hand.

In this paper, we study the structural properties of the DSEHS problem and then leverage them to develop a low-complexity RL algorithm based on value function approximation.
The proposed algorithm allows us to learn an accurate approximation of the optimal value function online, which enables in turn effective minimization of the packet queuing delay given the available harvested energy.
We demonstrate that the proposed algorithm achieves near optimal performance even when the learning updates are carried out intermittently. Moreover, competitive performance is demonstrated relative to a state-of-the-art learning algorithm, at potentially orders of magnitude lower computational complexity. Finally, our framework enables considerable performance gains over the widely used Q-learning algorithm.


\appendix

\textbf{Proof of Proposition~\ref{prop:well-behaved}.} 
The first condition in Definition~\ref{def:well-behaved} is satisfied by assumption. Thus, we only need to show that the PDS Learning algorithm satisfies the second and third conditions.

To simplify the proof, we first introduce some new notation. Using the PDS, we can factor the transition probabilities into known and unknown components, where the known component accounts for the transition from the current state to the PDS, i.e., $s\rightarrow\widetilde{s}$, and the unknown component accounts for the transition from the PDS to the next state, i.e., $\widetilde{s}\rightarrow s^\prime$~\cite{mastronarde2013joint}. Formally,
\begin{equation}\label{eq:factorized-tpf}
P(s'|s,a) = \sum_{\widetilde{s} \in \mathcal{S}} p_u(s'|\widetilde{s})p_k(\widetilde{s}|s,a),
\end{equation}
where the subscripts $k$ and $u$ denote the known and unknown components, respectively. We can factor the cost function similarly: 
\begin{equation}\label{eq:factorized-cost}
c(s,a) = c_k(s,a) + \sum_{\widetilde{s} \in \mathcal{S}}p_k(\widetilde{s}|s,a)c_u(\widetilde{s}).
\end{equation}
In our problem, the known and unknown costs and transition probabilities are defined as:
\begin{align}
c_{k}(s,a) &= b,\label{eq:known_cost} \\
c_{u}(\widetilde{s}) &= \eta\sum\nolimits_{l=0}^{\infty}P^{l}(l)\max(\widetilde{b}+l-N_b,0),\label{eq:unknown_cost} \\
\hspace{-0.1cm}P_{k}(\widetilde{s}|s,a) &= P^{f}(b-\widetilde{b}|a,h) \mathbb{I}_{\{\widetilde{e}=e-a \cdot e_{TX}\}} \mathbb{I}_{\{\widetilde{h}=h\}}, \label{eq:known_tpf} \\
P_{u}(s^\prime|\widetilde{s}) &= P^{l}(b^\prime-\widetilde{b}) P^{e_H}(e^\prime-\widetilde{e}) P^{h}(h^\prime|\widetilde{h}), \label{eq:unknown_tpf}
\end{align}
where $\mathbb{I}_{\{\cdot\}}$ is the indicator function. Note that \eqref{eq:unknown_tpf} is written for the case that $b^\prime < N_b$ and $e^\prime < N_e$. If $b^\prime = N_b$, for instance, then we should use $\sum_{l=N_b-\widetilde{b}}^{\infty}P^{l}(l)$ in \eqref{eq:unknown_tpf}. A similar modification is required for $e_{H}$ in the case that $e^\prime = N_e$. 
Using this new notation, we may rewrite the Bellman equations in~\eqref{eq:V_to_PDSV} and~\eqref{eq:PDSV_to_V} as follows:
\begin{align}
\widetilde{V}^{*}(\widetilde{s}) &= c_{u}(\widetilde{s}) + \gamma\sum_{s^\prime \in \mathcal{S}}P_{u}(s^\prime | \widetilde{s})V^{*}(s^\prime) \label{eq:PDSV_to_V_2} \\ 
V^{*}(s) &= \min_{a \in \mathcal{A}} \biggl\{c_{k}(s, a)+\sum_{\widetilde{s} \in \mathcal{S}} P_{k}(\widetilde{s} | s, a) \widetilde{V}^{*}(\widetilde{s})\biggr\} \label{eq:V_to_PDSV_2}
\end{align}

Plugging~\eqref{eq:V_to_PDSV_2} into~\eqref{eq:PDSV_to_V_2}, we can define a mapping $H_{PDS}$ that maps a $\widetilde{V}$-vector to a new $\widetilde{V}$-vector $H_{PDS}\widetilde{V}$ according to the formula
\begin{equation} \label{eq:H_PDS}
(H_{PDS} \widetilde{V})(\widetilde{s}) = c_{u}(\widetilde{s}) + \gamma \sum_{s^\prime \in \mathcal{S}} P_{u}(s^\prime | \widetilde{s}) \min_{a \in \mathcal{A}} \biggl\{c_{k}(s^\prime, a) + \sum_{\widetilde{s}^\prime \in \mathcal{S}} P_{k}(\widetilde{s}^\prime | s^\prime, a) \widetilde{V}(\widetilde{s}^\prime)\biggr\},
\end{equation}
where $\widetilde{s}$, $s^\prime$, and $\widetilde{s}^\prime$ denote the current PDS, next state, and next PDS, respectively.
Now, we can rewrite the PDS learning update in~\eqref{eq:pds-update} using the mapping $H_{PDS}$:
\begin{equation}\label{eq:pds-update-2}
\widetilde{V}^{n + 1}(\widetilde{s}^n) = (1 - \beta^n) \widetilde{V}^n(\widetilde{s}^n) + \beta^n \big[
(H_{PDS} \widetilde{V}^n)(\widetilde{s}^n) + w^n(\widetilde{s}^n)\big],
\end{equation}
where
\begin{equation}\label{eq:noise}
w^n(\widetilde{s}^n) = \eta \max(\widetilde{b}^n + l^n - N_b, 0) + \gamma V^{n}(s^{n + 1}) - \biggl[c_{u}(\widetilde{s}^n) + \gamma\sum_{s^\prime \in \mathcal{S}}P_{u}(s^\prime | \widetilde{s}^n)V^n(s^\prime)\biggr].
\end{equation}
For any history $F^n$, it is easy to show that $E[w^n(\widetilde{s}^n) | F^n] = 0$ and $|w^n(\widetilde{s}^n) | \leq V_{\max}$, where $V_{\max} = \max\{c(s, a)\} / (1 - \gamma)$. 

Now, we only need to show that the mapping $H_{PDS}$ satisfies the contraction property:
\begin{align*}
& \left|(H_{PDS}\widetilde{V})(\widetilde{s}) - \widetilde{V}^*(\widetilde{s})\right| \\
& = \gamma \sum_{s^\prime \in \mathcal{S}} P_{u}(s^\prime | \widetilde{s}) \left|V(s^\prime) - V^*(s^\prime)\right| \\
& = \gamma \sum_{s^\prime \in \mathcal{S}} P_{u}(s^\prime | \widetilde{s}) \left| \min_{a \in \mathcal{A}} \biggl\{c_k(s^\prime, a) + \sum_{\widetilde{s}^\prime \in \mathcal{S}} P_k(\widetilde{s}^\prime | s^\prime, a) \widetilde{V}(\widetilde{s}^\prime)\biggr\} - \min_{a \in \mathcal{A}} \biggl\{c_k(s^\prime, a) + \sum_{\widetilde{s}^\prime \in \mathcal{S}} P_k(\widetilde{s}^\prime | s^\prime, a) \widetilde{V}^{*}(\widetilde{s}^\prime)\biggr\}\right| \\
& \leq \gamma \sum_{s^\prime \in \mathcal{S}} P_{u}(s^\prime | \widetilde{s})  \max_{a \in \mathcal{A}} \left|\sum_{\widetilde{s}^\prime \in \mathcal{S}} P_k(\widetilde{s}^\prime | s^\prime, a) \widetilde{V}(\widetilde{s}^\prime) - \sum_{\widetilde{s}^\prime \in \mathcal{S}} P_k(\widetilde{s}^\prime | s^\prime, a) \widetilde{V}^*(\widetilde{s}^\prime)\right| \\
& = \gamma \sum_{s^\prime \in \mathcal{S}} P_{u}(s^\prime | \widetilde{s}) \max_{a \in \mathcal{A}} \sum_{\widetilde{s}^\prime \in \mathcal{S}} P_k(\widetilde{s}^\prime | s^\prime, a) \left|\left(\widetilde{V}(\widetilde{s}^\prime) - \widetilde{V}^*(\widetilde{s}^\prime)\right)\right| \\
& \leq \gamma \sum_{s^\prime \in \mathcal{S}} P_{u}(s^\prime | \widetilde{s}) \max_{a \in \mathcal{A}} \sum_{\widetilde{s}^\prime \in \mathcal{S}} P_k(\widetilde{s}^\prime | s^\prime, a) 
||{ \widetilde{V} - \widetilde{V}^*}|| \\
& = \gamma ||{\widetilde{V} - \widetilde{V}^*}||,
\end{align*}
where the first and second equalities follow by applying the definition of $(H_{PDS}\widetilde{V})(\widetilde{s})$ (see~\eqref{eq:H_PDS}); the first inequality follows from the fact that the difference of minimums is less than the maximum of differences; the third equality follows by rearranging terms; the final inequality follows by definition of the $L_\infty$ norm; and the last equality follows from the fact that $||{\widetilde{V} - \widetilde{V}^*}||$ does not depend on the summation variables $s^\prime$ and $\widetilde{s}^\prime$, and $P_{u}(s^\prime | \widetilde{s})$ and $P_k(\widetilde{s}^\prime | s^\prime, a)$ sum to 1.

\section{Conclusion}
\label{sec:con}

\balance
\bibliographystyle{IEEEtran}
\bibliography{refs}

\begin{thebibliography}{10}
\providecommand{\url}[1]{#1}
\csname url@samestyle\endcsname
\providecommand{\newblock}{\relax}
\providecommand{\bibinfo}[2]{#2}
\providecommand{\BIBentrySTDinterwordspacing}{\spaceskip=0pt\relax}
\providecommand{\BIBentryALTinterwordstretchfactor}{4}
\providecommand{\BIBentryALTinterwordspacing}{\spaceskip=\fontdimen2\font plus
\BIBentryALTinterwordstretchfactor\fontdimen3\font minus
  \fontdimen4\font\relax}
\providecommand{\BIBforeignlanguage}[2]{{%
\expandafter\ifx\csname l@#1\endcsname\relax
\typeout{** WARNING: IEEEtran.bst: No hyphenation pattern has been}%
\typeout{** loaded for the language `#1'. Using the pattern for}%
\typeout{** the default language instead.}%
\else
\language=\csname l@#1\endcsname
\fi
#2}}
\providecommand{\BIBdecl}{\relax}
\BIBdecl

\bibitem{Chakareski:15}
J.~Chakareski, ``Uplink scheduling of visual sensors: When view popularity
  matters,'' \emph{IEEE Trans. Commun.}, vol.~2, no.~63, pp. 510--519, Feb.
  2015.

\bibitem{Chakareski:11g}
------, ``Informative state-based video communication,'' \emph{IEEE Trans.
  Image Process.}, vol.~22, no.~6, pp. 2115--2127, Jun. 2013.

\bibitem{seyedi2010energy}
A.~Seyedi and B.~Sikdar, ``Energy efficient transmission strategies for body
  sensor networks with energy harvesting,'' \emph{IEEE Trans. Commun.},
  vol.~58, no.~7, pp. 2116--2126, 2010.

\bibitem{Chakareski:17}
J.~Chakareski, ``Aerial {UAV-IoT} sensing for ubiquitous immersive
  communication and virtual human teleportation,'' in \emph{Proc. IEEE INFOCOM
  Workshops}, Atlanta, GA, USA, May 2017.

\bibitem{Chakareski:17a}
------, ``Drone networks for virtual human teleportation,'' in \emph{Proc.
  MobiSys Workshops}, Niagara Falls, NY, USA, Jun. 2017.

\bibitem{zordandesign}
D.~Zordan, T.~Melodia, and M.~Rossi, ``On the design of temporal compression
  strategies for energy harvesting sensor networks,'' \emph{IEEE Trans.
  Wireless Commun.}, vol.~15, no.~2, pp. 1336--1352, Feb 2016.

\bibitem{kansal2007power}
A.~Kansal, J.~Hsu, S.~Zahedi, and M.~B. Srivastava, ``Power management in
  energy harvesting sensor networks,'' \emph{ACM Transactions on Embedded
  Computing Systems (TECS)}, vol.~6, no.~4, p.~32, 2007.

\bibitem{vullers2010energy}
R.~J. Vullers, R.~Van~Schaijk, H.~J. Visser, J.~Penders, and C.~Van~Hoof,
  ``Energy harvesting for autonomous wireless sensor networks,'' \emph{IEEE
  Solid-State Circuits Mag.}, vol.~2, no.~2, pp. 29--38, 2010.

\bibitem{gurakanenergy}
B.~Gurakan and S.~Ulukus, ``Energy harvesting multiple access channel with data
  arrivals,'' in \emph{IEEE GLOBECOM}, 2015.

\bibitem{lu2014dynamic}
X.~Lu, P.~Wang, D.~Niyato, and E.~Hossain, ``Dynamic spectrum access in
  cognitive radio networks with {RF} energy harvesting,'' \emph{IEEE Wireless
  Commun.}, vol.~21, no.~3, pp. 102--110, 2014.

\bibitem{sharma2010optimal}
V.~Sharma, U.~Mukherji, V.~Joseph, and S.~Gupta, ``Optimal energy management
  policies for energy harvesting sensor nodes,'' \emph{IEEE Trans. Wireless
  Commun.}, vol.~9, no.~4, 2010.

\bibitem{gunduz2014designing}
D.~Gunduz, K.~Stamatiou, N.~Michelusi, and M.~Zorzi, ``Designing intelligent
  energy harvesting communication systems,'' \emph{IEEE Commun. Mag.}, vol.~52,
  no.~1, pp. 210--216, 2014.

\bibitem{puterman2014markov}
M.~L. Puterman, \emph{Markov decision processes: discrete stochastic dynamic
  programming}.\hskip 1em plus 0.5em minus 0.4em\relax John Wiley \& Sons,
  2014.

\bibitem{ozel2011transmission}
O.~Ozel, K.~Tutuncuoglu, J.~Yang, S.~Ulukus, and A.~Yener, ``Transmission with
  energy harvesting nodes in fading wireless channels: Optimal policies,''
  \emph{IEEE J. Sel. Areas Commun.}, vol.~29, no.~8, pp. 1732--1743, 2011.

\bibitem{ho2012optimal}
C.~Ho and R.~Zhang, ``Optimal energy allocation for wireless communications
  with energy harvesting constraints,'' \emph{IEEE Trans. Signal Process.},
  vol.~60, no.~9, pp. 4808--4818, 2012.

\bibitem{yang2012optimal1}
J.~Yang and S.~Ulukus, ``Optimal packet scheduling in a multiple access channel
  with energy harvesting transmitters,'' \emph{Journal of Communications and
  Networks}, vol.~14, no.~2, pp. 140--150, 2012.

\bibitem{yang2012optimal}
------, ``Optimal packet scheduling in an energy harvesting communication
  system,'' \emph{IEEE Trans. Commun.}, vol.~60, no.~1, pp. 220--230, 2012.

\bibitem{michelusi2012optimal}
N.~Michelusi, K.~Stamatiou, and M.~Zorzi, ``On optimal transmission policies
  for energy harvesting devices,'' in \emph{Information Theory and Applications
  Workshop (ITA), 2012}.\hskip 1em plus 0.5em minus 0.4em\relax IEEE, 2012, pp.
  249--254.

\bibitem{aprem2013transmit}
A.~Aprem, C.~R. Murthy, and N.~B. Mehta, ``Transmit power control policies for
  energy harvesting sensors with retransmissions,'' \emph{IEEE J. Sel. Topics
  Signal Process.}, vol.~7, no.~5, pp. 895--906, 2013.

\bibitem{ho2010optimal}
C.~K. Ho and R.~Zhang, ``Optimal energy allocation for wireless communications
  powered by energy harvesters,'' in \emph{Proc. 2010 IEEE International
  Symposium on Information Theory (ISIT),}, 2010, pp. 2368--2372.

\bibitem{sutton1998reinforcement}
R.~Sutton and A.~Barto, \emph{Reinforcement learning: An introduction},
  1st~ed.\hskip 1em plus 0.5em minus 0.4em\relax MIT Press Cambridge, 1998.

\bibitem{mastronarde2013joint}
N.~Mastronarde and M.~van~der Schaar, ``Joint physical-layer and system-level
  power management for delay-sensitive wireless communications,'' \emph{IEEE
  Trans. Mobile Comput.}, vol.~12, no.~4, pp. 694--709, 2013.

\bibitem{blasco2013learning}
P.~Blasco, D.~Gunduz, and M.~Dohler, ``A learning theoretic approach to energy
  harvesting communication system optimization,'' \emph{IEEE Trans. Wireless
  Commun.}, vol.~12, no.~4, pp. 1872--1882, 2013.

\bibitem{watkins1992q}
C.~J. Watkins and P.~Dayan, ``Q-learning,'' \emph{Machine learning}, vol.~8,
  no. 3-4, pp. 279--292, 1992.

\bibitem{konda2000actor}
V.~R. Konda and J.~N. Tsitsiklis, ``Actor-critic algorithms,'' in \emph{NIPS},
  2000, pp. 1008--1014.

\bibitem{ortiz2016reinforcement}
A.~Ortiz, H.~Al-Shatri, X.~Li, T.~Weber, and A.~Klein, ``Reinforcement learning
  for energy harvesting point-to-point communications,'' in
  \emph{Communications (ICC), 2016 IEEE International Conference on}.\hskip 1em
  plus 0.5em minus 0.4em\relax IEEE, 2016, pp. 1--6.

\bibitem{xiao2015bayesian}
Y.~Xiao, Z.~Han, D.~Niyato, and C.~Yuen, ``Bayesian reinforcement learning for
  energy harvesting communication systems with uncertainty,'' in
  \emph{Communications (ICC), 2015 IEEE International Conference on}.\hskip 1em
  plus 0.5em minus 0.4em\relax IEEE, 2015, pp. 5398--5403.

\bibitem{pandana2005near}
C.~Pandana and K.~R. Liu, ``Near-optimal reinforcement learning framework for
  energy-aware sensor communications,'' \emph{IEEE J. Sel. Areas Commun.},
  vol.~23, no.~4, pp. 788--797, 2005.

\bibitem{liu2006rl}
Z.~Liu and I.~Elhanany, ``{RL-MAC}: A {QoS}-aware reinforcement learning based
  {MAC} protocol for wireless sensor networks,'' in \emph{Proc. 2006 IEEE
  International Conference on Networking, Sensing and Control}.\hskip 1em plus
  0.5em minus 0.4em\relax IEEE, 2006, pp. 768--773.

\bibitem{salodkar2008line}
N.~Salodkar, A.~Bhorkar, A.~Karandikar, and V.~Borkar, ``An on-line learning
  algorithm for energy efficient delay constrained scheduling over a fading
  channel,'' \emph{IEEE J. Sel. Areas Commun.}, vol.~26, no.~4, pp. 732--742,
  2008.

\bibitem{zhang1999finite}
Q.~Zhang and S.~A. Kassam, ``Finite-state markov model for rayleigh fading
  channels,'' \emph{IEEE Trans. Commun.}, vol.~47, no.~11, pp. 1688--1692,
  1999.

\bibitem{ngo2010monotonicity}
M.~H. Ngo and V.~Krishnamurthy, ``Monotonicity of constrained optimal
  transmission policies in correlated fading channels with {ARQ},'' \emph{IEEE
  Trans. on Signal Process.}, vol.~58, no.~1, pp. 438--451, 2010.

\bibitem{bertsekas1987data}
D.~P. Bertsekas, R.~G. Gallager, and P.~Humblet, \emph{Data networks}.\hskip
  1em plus 0.5em minus 0.4em\relax Prentice-hall Englewood Cliffs, NJ, 1987,
  vol.~2.

\bibitem{bertsekas1995neuro}
D.~P. Bertsekas and J.~N. Tsitsiklis, ``Neuro-dynamic programming: an
  overview,'' in \emph{Proc. 34th IEEE Conference on Decision and Control},
  vol.~1, 1995, pp. 560--564.

\bibitem{sharma2018structural}
N.~Sharma, N.~Mastronarde, and J.~Chakareski, ``Structural properties of
  optimal transmission policies for delay-sensitive energy harvesting wireless
  sensors,'' \emph{arXiv preprint arXiv:1803.09778}, 2018.

\bibitem{kaelbling1996reinforcement}
L.~Kaelbling, M.~Littman, and A.~Moore, ``Reinforcement learning: A survey,''
  \emph{Journal of artificial intelligence research}, pp. 237--285, 1996.

\end{thebibliography}

\end{document}